\documentclass[11pt,a4paper]{article}

\usepackage{amsmath,amsfonts,amssymb,amsthm}
\usepackage{mathbbol}
\usepackage{amsthm}
\usepackage{graphicx,color}
\usepackage{boxedminipage}
\usepackage[linesnumbered, boxed, algosection,noline]{algorithm2e}
\usepackage{framed}
\usepackage{thmtools}
\usepackage{thm-restate}
\usepackage{xspace}
\usepackage{vmargin}
\setmarginsrb{1in}{1in}{1in}{1in}{0pt}{0pt}{0pt}{6mm}
\usepackage{todonotes}
 \usepackage[pdftex, plainpages = false, pdfpagelabels, 
                 bookmarks=false,
                 bookmarksopen = true,
                 bookmarksnumbered = true,
                 breaklinks = true,
                 linktocpage,
                 pagebackref,
                 colorlinks = true,  
                 linkcolor = blue,
                 urlcolor  = blue,
                 citecolor = red,
                 anchorcolor = green,
                 hyperindex = true,
                 hyperfigures
                 ]{hyperref} 
 \usepackage{xifthen}
 \usepackage{tabularx}

 \usetikzlibrary{calc}

\newcommand{\pis}{\overrightarrow{\pi}}
\newcommand{\sigmas}{\overrightarrow{\sigma}}
\newcommand{\ds}{\overrightarrow{d}}
\newcommand{\iotas}{\overrightarrow{\iota}}
\newcommand{\rhos}{\overrightarrow{\rho}}

\DeclareMathOperator{\operatorClassNP}{{\sf NP}}
\newcommand{\classNP}{\ensuremath{\operatorClassNP}}

\DeclareMathOperator{\operatorClassFPT}{{\sf FPT}\xspace}
\newcommand{\classFPT}{\ensuremath{\operatorClassFPT}\xspace}

\newcommand{\Oh}{\mathcal{O}}

\newtheorem{theorem}{Theorem}
\newtheorem{lemma}{Lemma}
\newtheorem{claim}{Claim}
\newtheorem{corollary}{Corollary}

\newtheorem{observation}{Observation}
\newtheorem{proposition}{Proposition}
\newtheorem{redrule}{Reduction Rule}

\newcommand{\pname}{\textsc}
\newcommand{\ProblemFormat}[1]{\pname{#1}}
\newcommand{\ProblemIndex}[1]{\index{problem!\ProblemFormat{#1}}}
\newcommand{\ProblemName}[1]{\ProblemFormat{#1}\ProblemIndex{#1}{}\xspace}

 \newcommand{\probWS}{\ProblemName{Whitney Switches}}
 \newcommand{\probUWS}{\ProblemName{Unlabeled Whitney Switches}}
  \newcommand{\probSRD}{\ProblemName{Sequence Reversal Distance}}   
    \newcommand{\probRS}{\ProblemName{Sorting by Reversals}}

\newcommand{\NP}{{\ensuremath{\rm{NP}}}}



\newlength{\RoundedBoxWidth}
\newsavebox{\GrayRoundedBox}
\newenvironment{GrayBox}[1]%
   {\setlength{\RoundedBoxWidth}{.93\textwidth}
    \def\boxheading{#1}
    \begin{lrbox}{\GrayRoundedBox}
       \begin{minipage}{\RoundedBoxWidth}}%
   {   \end{minipage}
    \end{lrbox}
    \begin{center}
    \begin{tikzpicture}%
       \node(Text)[draw=black!20,fill=white,rounded corners,%
             inner sep=2ex,text width=\RoundedBoxWidth]%
             {\usebox{\GrayRoundedBox}};
        \coordinate(x) at (current bounding box.north west);
        \node [draw=white,rectangle,inner sep=3pt,anchor=north west,fill=white] 
        at ($(x)+(6pt,.75em)$) {\boxheading};
    \end{tikzpicture}
    \end{center}}     

\newenvironment{defproblemx}[2][]{\noindent\ignorespaces%
                                \FrameSep=6pt%
                                \parindent=0pt%
                \vspace*{-1.5em}
                \ifthenelse{\isempty{#1}}{%
                  \begin{GrayBox}{\textsc{#2}}%                
                }{%
                  \begin{GrayBox}{\textsc{#2}  parameterized by~{#1}}%  
                }
                \begin{tabular*}{\textwidth}{@{\hspace{.1em}} >{\itshape} p{1.8cm} p{0.8\textwidth} @{}}%        
            }{
                \end{tabular*}%
                \end{GrayBox}%
                \ignorespacesafterend
            }

\newcommand{\defproblema}[3]{% FJR Version
  \begin{defproblemx}{#1}
    Input:  & #2 \\
    Task: & #3
  \end{defproblemx}
}%

\begin{document}
\title{Kernelization of Whitney Switches
\thanks{The research leading to these results have  been supported by the Research Council of Norway via the projects ``MULTIVAL" (grant no. 263317).}}

\author{
Fedor V. Fomin\thanks{
Department of Informatics, University of Bergen, Norway.} \addtocounter{footnote}{-1}
\and
Petr A. Golovach\footnotemark{}
}

\date{}

\maketitle

\begin{abstract}
A fundamental theorem of Whitney from 1933 asserts that $2$-connected graphs $G$ and $H$ are $2$-isomorphic, or equivalently, their cycle matroids are isomorphic, if and only if $G$ can be transformed into $H$ by a series of operations called Whitney switches.
In this paper we 
consider the quantitative 
question arising from 
 Whitney's theorem: Given two 2-isomorphic graphs, can we transform one into another by applying at most $k$ Whitney switches? 
 This problem is already \classNP-complete for cycles, and we investigate its parameterized complexity. 
 We show that the problem admits a kernel of size $\Oh(k)$, and thus, is fixed-parameter tractable when parameterized by $k$.

\end{abstract}

\section{Introduction}\label{sec:intro}
A fundamental result of 
Whitney from 1933~\cite{Whitney33}, asserts that every $2$-connected graph is completely characterized, up to a series of Whitney switches  (also known as $2$-switches), by  its edge set and  cycles.  This theorem is one of the cornerstones of Matroid Theory, since it provides an exact characterization of two graphs having isomorphic cycle matroids \cite{VertiganW97}. In graph drawing and graph embeddings, this theorem (applied to dual graphs) is used to  characterize all drawings  of a planar graph on the palne~\cite{ChimaniHM12}.

A \emph{Whitney switch}  is an operation that from  a 2-connected graph $G$, constructs graph $G'$ as follows. 
Let $\{u,v\}$ be two vertices of $G$ 
whose removal separates $G$ into two disjoint subgraphs $G_1$ and $G_2$. 
The graph $G'$ is obtained by 
flipping the neighbors of $u$ and $v$ in the set of vertices of $G_2$. In other words, for every vertex $w\in V(G_2)$, if $w$ was adjacent to $u$ in $G$, in graph $G'$  edge $uw$  is replaced by $vw$. Similarly, if  $w$ was adjacent to $v$ in $G$, we replace $vw$  by $uw$. See  Figure~\ref{fig:switching} for an example.

\begin{figure}[ht]
\centering
\scalebox{0.7}{
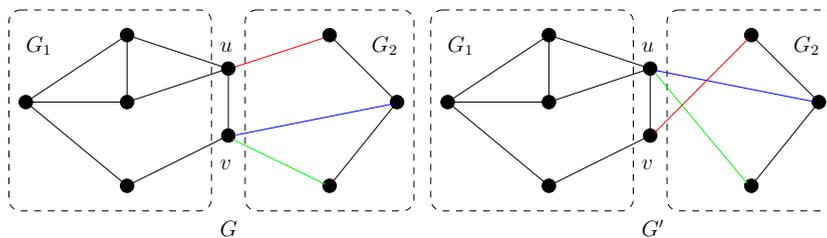}
\caption{$G'$ is obtained from $G$ by the Whitney switch with respect to the partition of $G-\{u,v\}$ into $G_1$ and $G_2$.}
\label{fig:switching}
\end{figure}

If we view the graph $G$ as a graph with labelled edges, then a Whitney switch transforms $G$    into a graph  $G'$ with the same set of labelled edges, however graphs $G$ and $G'$ are  not necessarily isomorphic. On other hand, graphs $G$ and $G'$ have the same set of cycles in the following sense: a set of (labelled) edges  forms a cycle  in $G$ if and only if it forms a cycle in $G'$. (In other words, the cycle matroids of $G$ and $G'$ are isomorphic.) What Whitney's theorem says that the opposite is also true: if there is a cycle-preserving mapping between graphs $G$ and $G'$ then one graph can be transformed into another  by a sequence of Whitney switches.  To state the theorem of Whitney more precisely, we need to define $2$-isomorphisms. 

We say that 2-connected graphs $G$ and $H$ are  \emph{$2$-isomorphic} if there is a bijection $\varphi\colon E(G)\rightarrow E(H)$ such that  $\varphi$ and $\varphi^{-1}$ preserve cycles, that is, for every cycle $C$ of $G$, $C$ is mapped to a cycle of $H$  by $\varphi$ and, symmetrically, every cycle of $H$ is mapped to a cycle of $G$ by $\varphi^{-1}$.   We refer to $\varphi$ as to  \emph{$2$-isomorphism} from $G$ to $H$. 
An isomorphism $\psi\colon V(G)\rightarrow V(H)$ is a \emph{$\varphi$-isomorphism} if for every edge $uv\in E(G)$, $\varphi(uv)=\psi(u)\psi(v)$, and $G$ and $H$ are $\varphi$-isomorphic if there is an isomorphism $G$ to $H$ that is a $\varphi$-isomorphism.  Let us note that for $3$-connected graphs the notions of $2$-isomorphism and  {$\varphi$-isomorphism} coincides. More precisely, if $G$ is $3$-connected and $2$-isomorphic to $H$ under $\varphi$ then $G$ and $H$ are $\varphi$-isomorphic
\cite[Lemma~1]{Truemper80}. But for $2$-connected graphs this is not true. For example, the graphs in Fig.~\ref{fig:switching} are not isomorphic but are 2-isomorphic. But even isomorphic graphs with $2$-isomorphism $\varphi$ not always have a {$\varphi$-isomorphism}. For example,  for the $2$-isomorphism $\varphi$ in Fig.~\ref{fig:cycles} mapping a cycle $G$ into another cycle $H$ (we view these cycles as labelled graphs), there is no $\varphi$-isomorphism.  (For every $\varphi$-isomorphism   edges $\varphi(a)$ and $\varphi(b)$ should have an endpoint in common.)
On the other hand, graph $G'$ obtained from $G$ by Whitney switch  (for vertices $u$ and $v$)  is $\varphi$-isomorphic to $H$.

\begin{figure}[ht]
\centering
\scalebox{0.3}{
\includegraphics{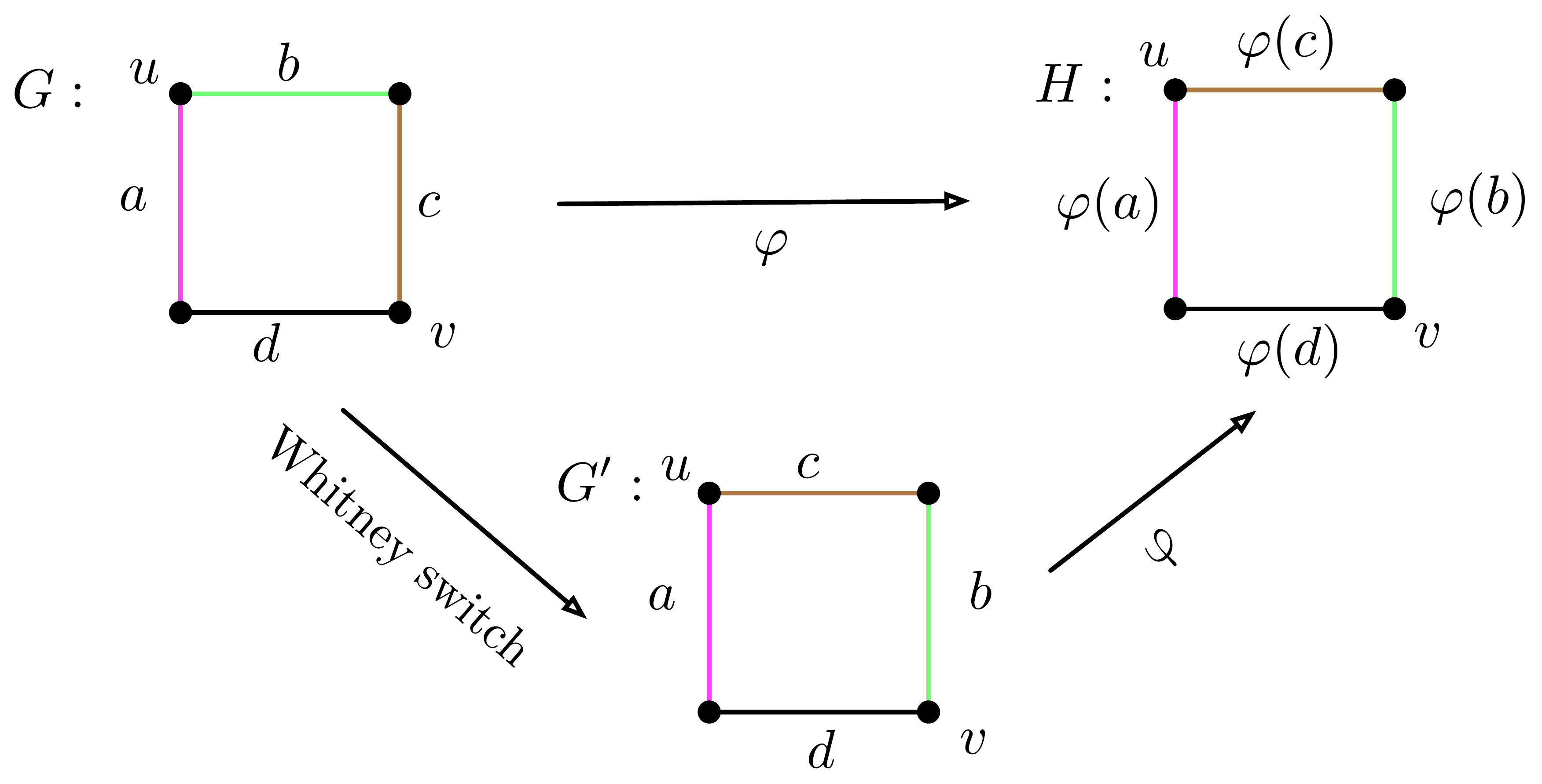}}
\caption{Graph $G$ is not $\varphi$-isomorphic to $H$ but its Whitney switch $G'$ is. }
\label{fig:cycles}
\end{figure}

\begin{theorem}[\textbf{Whitney's theorem}~\cite{Whitney33}]\label{thm:Withney}
If   there is a $2$-isomorphism $\varphi$ from graph  $G$ to graph $H$, then  $G$ can be transformed by a sequence 
of Whitney switches to  a graph $G'$ which is $\varphi$-isomorphic to $H$. 
\end{theorem}

However, Whitney's theorem does not provide an answer to the following  computational question: Given a $2$-isomorphism $\varphi$ from graph  $G$ to graph $H$, \emph{what is the minimum number of Whitney switches required to transform $G$ to a graph  $\varphi$-isomorphic to $H$?} Truemper in~\cite{Truemper80} proved that $n-2$ switches always suffices, where $n$ is the number of vertices in $G$. He also   proved that this upper bound it tight, that is, there are graphs $G$ and $H$ for which $n-2$ switches are necessary. 
In this paper we study the algorithmic complexity of the following problem about Whitney switches.

\defproblema{\probWS}{$2$-Isomorphic $n$-vertex graphs $G$ and $H$ with a $2$-isomorphism $\varphi\colon E(G)\rightarrow E(H)$, and a nonnegative integer $k$. }{Decide whether it is possible to obtain from $G$ a graph $G'$ that is  $\varphi$-isomorphic to $H$  by at most $k$ Whitney switches.}

The departure  point for our study of  \probWS is an easy reduction (Theorem~\ref{cor:hard})  from \probRS that establishes 
 \classNP-completeness of \probWS  even when input graphs $G$ and $H$ are  cycles. Our main algorithmic result
 is the following theorem (we postpone the definition of a kernel  till Section~\ref{sec:prelim}).  
 Informally, it means that the instance of the problem can be compressed in polynomial time to an equivalent instance with two graphs on $\Oh(k)$ vertices.  It also implies that  \probWS  is fixed-parameter tractable parameterized by $k$.

\begin{theorem}\label{thm:main}    
\probWS  admits a kernel with $\Oh(k)$ vertices and is solvable in  $2^{\Oh(k\log k)}\cdot n^{\Oh(1)}$ time.
\end{theorem}

\medskip

While Theorem~\ref{thm:main}    is not restricted to planar graphs, pipelined with the well-known connection of planar embeddings and Whitney switches, it can be used to obtain interesting algorithmic consequences about distance between planar embeddings of a graph. Recall that 
graphs $G$ and $G^*$ are called abstractly dual if there is a bijection $\pi \colon E(G)\to E(G^*)$ such that edge set $E\subseteq E(G)$ forms a cycle in $G$ if and only if $\pi(E)$ is a minimal edge-cut in $G^*$. By another classical theorem of Whitney \cite{MR1501641}, a graph $G$ has a dual graph if and only if $G$ is planar. Moreover, an embedding of a planar graph into a sphere is uniquely defined by the planar graph $G$ and edges of the faces, or equivalently, its dual graph $G^*$.  
While every 
3-connected planar graph has a unique embedding into the sphere,  a $2$-connected  graph
can have several non-equivalent embeddings, and hence several non-isomorphic dual graphs.  If $G_1^*$ and $G_2^*$ are dual graphs of graph $G$, then $G_1^*$ is 2-isomorphic to $G_2^*$. Then by Theorem~\ref{thm:Withney}, by a sequence of  Whitney switches  $G_1^*$ can be transformed into  $G_2^*$, or equivalently, the  embedding of $G$ corresponding to $G_1^*$ can be transformed  to embedding of $G$ corresponding to $G_2^*$. We refer to the survey of Carsten Thomassen  \cite[Section~2.2]{MR1373661} for more details. By Theorem~\ref{thm:main}, we  have  that given two planar embeddings of a (labelled) 2-connected graph $G$, deciding whether one embedding    can be transformed into another by making use of at most $k$ Whitney switches, admits a kernel of size $\Oh(k)$ and is fixed-parameter tractable.

 \medskip
\noindent\textbf{Related work.} Whitney's theorem had a strong impact on the development of modern graph and matroid theories. While the original proof is long, a number of simpler proofs are known in the literature. The most relevant to our work is the proof of 
Truemper in~\cite{Truemper80}, whose proof is on the application of Tutte decomposition~\cite{Tutte66,Tutte84}.

 \probWS can be seen as an example of reconfiguration problems. The study of reconfiguration problems becomes a popular trend in parameterized complexity, see  e.g. \cite{MouawadN0SS17,lokshtanov2018reconfiguration}.

 The well-studied problem which is similar in spirit to \probWS is the problem of computing the  flip distance for triangulations of of a set of points. 
 The parameterized complexity of this problem was studied in  \cite{MR2541971,MR2667389}.  As we mentioned above, 
 \probWS for planar graphs is equivalent to the problem of computing the Whitney switch distance between planar embeddings. 
 We refer to the survey of Bose and
Hurtado~\cite{BoseH09} for the discussion of the relations between geometric and graph variants. The problem is known to 
be  \NP-complete~\cite{LubiwP15} and \classFPT parameterized by the number of flips~\cite{KanjSX17}. For the special case when the set of points defines a convex polygon, the problem of computing the flip distance between triangulations is  equivalent to computing  the rotation distance between two binary trees. For that case linear kernels are known 
\cite{MR2541971,MR2667389} but for the general case the existence of a polynomial kernel is open.

\medskip
\noindent\textbf{Overview of the proof of Theorem~\ref{thm:main}.}  The main tool in the construction of the kernel  is 
 the classical Tutte decompositions~\cite{Tutte66,Tutte84}.
  We postpone the formal definition till Section~\ref{sec:prelim}, informally,  the Tutte decomposition of a 2-connected graph represents the vertex separators of size two in a tree-like structure. Each node of this tree represents a part of the graph (or \emph{bag}) that is either a 3-connected graph or a cycle, and each edge corresponds to a separator of size two. Then a 2-isomorphism of $G$ and $H$  allows to establish an isomorphism of the trees representing the Tutte decompositions of the input graphs.  
After that,  potential Whitney switches can be divided into two types: the switches with respect to separators corresponding to the edges of the trees and the switches with respect to separators formed by nonadjacent vertices of a cycle-bag. The switches of the first type are relatively easy to analyze and we can identify necessary switches of this type. The ``troublemakers" that make the problem hard are switches of the second type. To deal with them, we use the structural results about sorting of permutations by reversals of 
Hannenhalli and Pevzner~\cite{HannenhalliP96} adapted for our purposes. This allows us to identify a set of vertices of size $\Oh(k)$ that potentially can be used for Whitney switches transforming $G$ to $H$.   Given such a set of crucial vertices, we simplify the structure of the input graphs and then reduce their size.

\medskip
\noindent{\bf Organization of the paper.}
In Section~\ref{sec:prelim}, we give basic definitions. In Section~\ref{sec:reversals}, we discuss the \probRS problem for permutations that is closely related  to \probWS. Section~\ref{sec:techn} contains structural results used by our kernelization algorithm, and in Section~\ref{sec:kern}, we give the algorithm itself. We conclude in Section~\ref{sec:concl} by discussing further directions of research.

\section{Preliminaries}\label{sec:prelim}  

\medskip
\noindent{\bf Graphs.} 
All graphs considered in this paper are finite undirected graphs without loops or multiple edges, unless it is specified explicitly that we consider directed graphs (in Section~\ref{sec:concl} we deal with tournaments).  
We follow the standard graph theoretic notation and terminology (see, e.g., \cite{Diestel12}). For each of the graph problems considered in this paper, we let $n=|V(G)|$ and $m=|E(G)|$ denote the number of vertices and edges,
respectively, of the input graph $G$ if it does not create confusion. 
For a graph $G$ and a subset $X\subseteq V(G)$ of vertices, we write $G[X]$ to denote the subgraph of $G$ induced by $X$.
For a set of vertices $S$, $G-S$ denotes the graph obtained by deleting the vertices of $S$, that is, $G-S=G[V(G)\setminus S]$; for a vertex $v$, we write $G-v$ instead of $G-\{v\}$.
Similarly, for a set of edges $A$ (an edge $e$, respectively), $G-A$ ($G-e$, respectively) denotes the graph obtained by the deletion of the edges of $A$ (an edge $e$, respectively). 
For a vertex $v$, we denote by $N_G(v)$ the \emph{(open) neighborhood} of $v$, i.e., the set of vertices that are adjacent to $v$ in $G$ and we use $E_G(v)$ to denote the set of edges incident to $v$. We use $N_G[v]$ to denote the \emph{closed neighborhood}, that is $N_G(v)\cup \{v\}$.
For $S\subseteq V(G)$, $N_G[S]=\bigcup_{v\in S}N_G[v]$ and $N_G(S)=N_G[S]\setminus S$. We write $N_G^2(v)=N_G(N_G[v])$ for a vertex $v$ to denote the \emph{second neighborhood}.
A vertex $v$ is \emph{simplicial} if $N_G(v)$ is a \emph{clique}, that is, a set of pairwise adjacent vertices. 
A pair $(A,B)$, where $A,B\subseteq V(G)$, is a \emph{separation} of $G$ if $A\cup B=V(G)$,  $A\setminus B\neq\emptyset$, $B\setminus A\neq\emptyset$ and $G$ has no edge $uv$ with $u\in A\setminus B$ and $v\in B\setminus A$;
$|A\cap B|$ is the \emph{order} of the separation. If the order is 2, then we say that $(A,B)$ is a \emph{Whitney separation}.
 A set $S\subseteq V(G)$ is a \emph{separator} of $G$ if there is a separation $(A,B)$ of $G$
 with $S=A\cap B$. For a positive integer $k$, a graph $G$ is \emph{$k$-connected} if $G$ is a connected graph with at least $k+1$ vertices without a separator of size at most $k-1$. In particular, $G$ is 2-connected if $G-v$ is connected for every $v\in V(G)$.

\medskip
\noindent{\bf Isomorphisms.} Graphs $G$ and $H$ are \emph{isomorphic} if there is bijection
 $\eta \colon V(G)\rightarrow V(H)$, called \emph{isomorphism}, preserving edges, that is,   
  $uv\in E(G)$ if and only if  $ \eta(u)\eta(v)\in E(H).$
We say that 2-connected graphs $G$ and $H$ are  \emph{$2$-isomorphic} if there is a bijection $\varphi\colon E(G)\rightarrow E(H)$ such that  $\varphi$ and $\varphi^{-1}$ preserve cycles, that is, for every cycle $C$ of $G$, $C$ is mapped to a cycle of $H$  by $\varphi$ and, symmetrically, every cycle of $H$ is mapped to a cycle of $G$ by $\varphi^{-1}$.   We refer to $\varphi$ as to  \emph{$2$-isomorphism} from $G$ to $H$. 
An isomorphism $\psi\colon V(G)\rightarrow V(H)$ is a \emph{$\varphi$-isomorphism} if for every edge $uv\in E(G)$, $\varphi(uv)=\psi(u)\psi(v)$, and $G$ and $H$ are $\varphi$-isomorphic if there is an isomorphism $G$ to $H$ that is a $\varphi$-isomorphism.   

\medskip
\noindent{\bf Whitney switches.} 
Let $G$ be a 2-connected graph. Let also $(A,B)$ be a Whitney separation of $G$ with $A\cap B=\{u,v\}$. 
 The \emph{Whitney switch}  operation  with respect to $(A,B)$ transforms $G$ as follows: take $G[A]$ and $G[B]$ and identify the vertex $u$ of $G[A]$ with the vertex $v$ of $G[B]$ and, symmetrically, $v$ of $G[A]$ with $u$ of $G[B]$; if $u$ and $v$ are adjacent in $G$, then the edges $uv$ of $G[A]$ and $G[B]$ are identified as well. 
Let $G'$ be the obtained graph.
We define the mapping $\sigma_{(A,B)}\colon E(G)\rightarrow E(G')$ that maps the edges of $G[A]$ and $G[B]$, respectively, to themselves. It is easy to see that 
$\sigma_{(A,B)}$ is a 2-isomorphism of $G$ to $G'$. Therefore, if $\varphi$ is a 2-isomorphism of $G$ to a graph $H$, then $\varphi\circ \sigma_{(A,B)}^{-1}$ is a 2-isomorphism of $G'$ to $H$. To simplify notation, we assume, if it does not create confusion, that the sets of edges of $G$ and $G'$ are identical and we only change incidences by switching. In particular, under this assumption,
we have that $\varphi\circ \sigma_{(A,B)}^{-1}=\varphi$. We also assume that the graphs $G$ and $G'$ have the same sets of vertices.

\medskip
\noindent{\bf Tutte decomposition.} Our kernelization algorithm for \probWS is based on the classical result of Tutte~\cite{Tutte66,Tutte84} about decomposing of 2-connected graphs via separators of size two. 
 Following Courcelle~\cite{Courcelle99}, we define Tutte decompositions in the terms of tree decompositions. 
 
A   {\em tree decomposition} of
a graph $G$ is a pair $\mathcal{T}=(T,\{X_t\}_{t\in V(T)})$, where $T$ is a tree whose every node $t$ is assigned a vertex subset $X_t\subseteq V(G)$, called a bag,
such that the following three conditions hold: 
\begin{description}
\item[(T1)] $\bigcup_{t\in V(T)} X_t =V(G)$, that is, every vertex of $G$ is in at least one bag,
\item[(T2)] for every $uv\in E(G)$, there exists a node $t$ of $T$ such that the bag $X_t$ contains both $u$ and $v$,
\item[(T3)] for every  $v\in V(G)$, the set $T_v = \{t\in V(T) \mid v\in X_t\}$, i.e., the set of nodes whose corresponding  bags contain $v$, induces
a connected subtree of $T$.
\end{description}
To distinguish between the vertices of the decomposition tree $T$ and the vertices of
the graph $G$, we will refer to the vertices of $T$ as {\em{nodes}}. 

Let $\mathcal{T}=(T,\{X_t\}_{t\in V(T)})$ be a tree decomposition of $G$. The \emph{torso} of $X_t$ for $t\in V(T)$ is the graph obtained from $G[X_t]$ by additionally making adjacent every two distinct vertices $u,v\in X_t$ such that there is $t'\in V(T)$ adjacent to $t$ with $u,v\in X_t\cap X_{t'}$. 
For adjacent $t,t'\in V(T)$, $X_t\cap X_{t'}$ is the \emph{adhesion set} of the bags $X_t$ and $X_{t'}$ and $|X_t\cap X_{t'}|$ is the \emph{adhesion} of the bags. The maximum adhesion of adjacent bags is called the \emph{adhesion} of the tree decomposition.  

Let $G$ be a 2-connected graph. A tree decomposition $\mathcal{T}=(T,\{X_t\}_{t\in V(T)})$ is said to be a \emph{Tutte decomposition} if $\mathcal{T}$ is a tree decomposition of adhesion 2 
such that there is a partition $(W_2,W_{\geq 3})$ of $V(T)$ such that the following holds:
\begin{description}
\item[(T4)] $|X_t|=2$ for  $t\in W_2$ and $|X_t|\geq 3$ for $t\in W_{\geq 3}$,
\item[(T5)] the torso of $X_t$ is either a 3-connected graph or a cycle for every $t\in W_{\geq 3}$,
\item[(T6)] for every $t\in W_2$, $d_T(t)\geq 2$ and $t'\in W_{\geq 3}$ for each neighbor $t'$ of $t$,
\item[(T7)] for every $t\in W_{\geq 3}$, $t'\in W_2$ for each neighbor $t'$ of $t$,
\item[(T8)] if $t\in W_2$ and $d_T(t)=2$, then for the neighbors $t'$ and $t''$ of $t$, either the torso of $t'$ or the torso of $t''$ is a 3-connected graph or the vertices of $X_t$ are adjacent in $G$. 
\end{description}

Notice that the bags $X_t$ for $t\in W_2$ are distinct separators of $G$ of size two, and $X_t\subseteq X_{t'}$ for $t\in W_2$ and $t'\in N_T(t)$. Observe also that if $\{u,v\}$ is a separator of $G$ of size two, then either $\{u,v\}=X_t$ for some $t\in W_2$ or $u,v\in X_t$ for $t\in W_{\geq 3}$ such that  the torso of $X_t$ is a cycle and $u$ and $v$ are nonadjacent vertices of the torso.

Combining the results of Tutte~\cite{Tutte66,Tutte84}  and of Hopcroft and Tarjan~\cite{HopcroftT73}, we state the following proposition.
 
\begin{proposition}[\cite{Tutte66,Tutte84,HopcroftT73}]\label{prop:tutte}
A 2-connected graph $G$ has a unique Tutte decomposition that can be constructed in linear time.
\end{proposition}

\noindent
{\bf Parameterized Complexity and Kernelization.}
We refer to the  books~\cite{CyganFKLMPPS15,DowneyF13,FominLSM19} for the detailed introduction to the field. Here we only give the most basic definitions. 
In the Parameterized Complexity theorey, the computational complexity  is measured as a function of the input size $n$ of a problem and an integer \emph{parameter} $k$ associated with the input.
A parameterized problem is said to be \emph{fixed parameter tractable} (or \classFPT) if it can be solved in time $f(k)\cdot n^{\Oh(1)}$ for some function~$f$. 
A \emph{kernelization } algorithm for a parameterized problem $\Pi$ is a polynomial algorithm that maps each instance $(I,k)$ of $\Pi$ to an instance $(I',k')$ of $\Pi$ such that
\begin{itemize}
\item[(i)] $(I,k)$ is a yes-instance of $\Pi$ if and only if $(I',k')$ is a yes-instance of $\Pi$, and
\item[(ii)] $|I'|+k'$ is bounded by~$f(k)$ for a computable function~$f$.
\end{itemize}
Respectively, $(I',k')$ is a \emph{kernel} and $f$ is its \emph{size}. A kernel is \emph{polynomial} if $f$ is polynomial.
It is common to present a kernelization algorithm as a series of \emph{reduction rules}.
A reduction rule for a parameterized problem is an algorithm that takes an instance of the problem and computes in polynomial time another instance that is more ``simple'' in a certain way.
A reduction rule is \emph{safe} if the computed instance is equivalent to the input instance.

\section{Sorting by reversals}\label{sec:reversals}
Sorting by reversals is the classical problem with many applications including bioinformatics. 
We refer to the book of Pevzner~\cite{Pevzner00} for the detailed survey of results and applications of this problem.
This problem is also strongly related to   \probWS---solving the problem for two cycles is basically the same as sorting 
   circular permutations by reversals. First we use this relation to observe the NP-completeness. But we also need  to establish some structural properties of  sorting by reversals which will be used in kernelization algorithm. 
   
   \medskip

Let $\pi=(\pi_1,\ldots,\pi_n)$ be a permutation of $\{1,\ldots,n\}$, that is, a bijective mapping of $\{1,\ldots,n\}$ to itself. Throughout this section, all considered permutations are permutations of  $\{1,\ldots,n\}$.
For $1\leq i\leq j\leq n$, the \emph{reversal} $\rho(i,j)$ reverse the order of elements $\pi_i,\ldots,\pi_j$ and transforms $\pi$ into $$\rho(i,j)\circ \pi=(\pi_1,\ldots,\pi_{i-1},\pi_j,\pi_{j-1},\ldots,\pi_i,\pi_{j+1},\ldots,\pi_n).$$ 
The \emph{reversal distance} $d(\pi,\sigma)$ between two permutations $\pi$ and $\sigma$
is the minimum number of reversals needed to transform $\pi$ to $\sigma$.  For a permutation $\pi$, $d(\pi)=d(\pi,\iota)$, where 
$\iota$ is the identity permutation; note that $d(\pi,\sigma)=d(\sigma^{-1}\circ\pi,\iota)$ and this means that computing the reversal distance can be reduced to sorting a permutation by the minimum number of reversals. 

These definitions can be extended for circular permutations (further, we may refer to usual permutations as \emph{linear} to avoid confusion). We say that  $\pi^c=(\pi_1,\ldots,\pi_n)$ is a \emph{circular} permutation if $\pi^c$ is the class of the permutations that can be obtained from the linear permutation $(\pi_1,\ldots,\pi_n)$ by \emph{rotations} and \emph{reflections}, that is, all the permutations 
$$(\pi_1,\ldots,\pi_n),(\pi_n,\pi_1,\ldots,\pi_{n-1}),\ldots,(\pi_2,\ldots,\pi_n,\pi_1)$$ and 
$$(\pi_n,\ldots,\pi_1),(\pi_1,\pi_n,\ldots,\pi_2),\ldots,(\pi_{n-1},\ldots,\pi_1,\pi_n)$$ 
composing one class  are identified, meaning that we do not distinguish them when discussing circular permutations.
For $i,j\in\{1,\ldots,n\}$, the \emph{circular reversal}  $\rho^c(i,j)$ is defined in the same way as $\rho(i,j)$ if $i\leq j$ and for $i>j$, $\rho^c(i,j)$ transforms 
$\pi^c$ into
$$\rho^c(i,j)\circ \pi^c=(\pi_n,\pi_{n-1},\ldots,\pi_{i},\pi_{j+1},\ldots,\pi_{i-1},\pi_{j},\pi_{j-1}\ldots,\pi_1).$$ 
The \emph{circular reversal distance} $d^c(\pi^c,\sigma^c)$ and $d^c(\pi^c)$ are defined in the same way as for linear permutations.

\begin{figure}[ht]
\centering
\scalebox{0.6}{
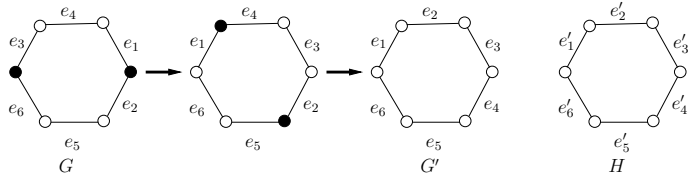}
\caption{The construction of $G'$ that is $\varphi$-isomorphic to $H$ by the Whitney switches 
corresponding to the sorting by reversals $(3,4,1,2,5,6)\rightarrow (1,4,3,2,5,6)\rightarrow(1,2,3,4,5,6)$; $\varphi(e_i)=e_i'$ for $i\in\{1,\ldots,6\}$, the vertices of the separators for the switches are shown in black.} 
\label{fig:reverse}
\end{figure}

To see the connection between Whitney switches and circular reversals of permutations, consider a cycle $G$ with the vertices $v_1,\ldots,v_n$ for $n\geq 4$ taken in the cycle order and the edges $e_i=v_{i-1}v_i$ for $i\in\{1,\ldots,r\}$ assuming that $v_0=v_n$.  Let $1\leq i<j\leq n$ be such that $v_i$ and $v_j$ are not adjacent. 
Then the Whitney switch with respect to $(A,B)$, where $A=\{v_1,\ldots,v_i\}\cup\{v_j,\ldots,v_n\}$ and $B=\{v_i,\ldots,v_j\}$ is equivalent to applying the reversal $\rho^c(i+1,j)$ to the circular permutation $(e_1,\ldots,e_n)$ of the edges of $G$. Moreover, let $H$ be a cycle with $n$ vertices and denote by $e_1',\ldots,e_n'$ its edges in the cycle order.
Notice that every bijection $\varphi\colon E(G)\rightarrow E(H)$ is a 2-isomorphism of $G$ to $H$, and $G$ and $H$ are $\varphi$-isomorphic if and only if  
the circular permutation $\pi^c=(\varphi^{-1}(e_1'),\ldots,\varphi^{-1}(e_n'))$ is the same as $\sigma^c=(e_1,\ldots,e_n)$.  Clearly, we can assume that $\pi^c$ is a permutation of $\{1,\ldots,n\}$ and $\sigma^c$ is the identity permutation. Then  $G$ can be transformed to a graph $G'$ $\varphi$-isomorphic to $H$ by at most $k$ Whitney switches if and only if $d^c(\pi^c)\leq k$. An example is shown in Fig.~\ref{fig:reverse}.

 In particular, the above observation implies the hardness of \probWS, because the computing of the reversal distances is known to be \classNP-hard. For linear permutations, it was shown by Caprara in~\cite{Caprara97}. The following result for circular permutations was obtained by 
Solomon, Sutcliffe, and Lister~\cite{SolomonSL03}.

\begin{proposition}[\cite{SolomonSL03}]\label{prop:perm-hard}
It is \classNP-complete to decide, given a circular permutation $\pi^c$ and a nonnegative integer $k$, whether $d^c(\pi^c)\leq k$.  
 \end{proposition}

This  brings us to  the following corollary.

\begin{theorem}\label{cor:hard}
\probWS is \classNP-complete even when restricted to cycles.
\end{theorem}

For our kernelization algorithm, we need some further structural results about reversals in an optimal sorting sequence.

Let $\pi=(\pi_1,\ldots,\pi_n)$ be a linear permutation. For $1\leq i\leq j\leq n$, we say that $(\pi_i,\ldots,\pi_j)$ is an \emph{interval} of $\pi$. An interval $(\pi_i,\ldots,\pi_j)$ is called a \emph{block} if either $i=j$ or $i<j$ and for every $h\in\{i+1,\ldots,j\}$, $|\pi_{h-1}-\pi_h|=1$, that is, a block is formed by consecutive integers in $\pi$ in either the ascending or descending order. An inclusion maximal block is called a \emph{strip}. In other words, a strip is an inclusion maximal interval that has no \emph{breakpoint}, that is, a pair of elements $\pi_{h-1},\pi_h$ with $|\pi_{h-1}-\pi_h|\geq 2$. It is said that a reversal $\rho(p,q)$ \emph{cuts} a strip $(\pi_i,\ldots,\pi_j)$ if  either $i<p\leq j$ or $i\leq q<j$, that is, the reversals separates elements that are consecutive in the identity permutation. 

It is known that there are cases when every optimal sorting by reversal requires a reversal that cuts a strip. For example, as was pointed by Hannenhalli and Pevzner in~\cite{HannenhalliP96}, the permutation $(3,4,1,2)$ requires three reversals that do not cut strips, but the sorting can be done by two reversals: \footnote{This example can be extended for circular permutations:
$(3,4,1,2,5,6)\rightarrow (1,4,3,2,5,6)\rightarrow(1,2,3,4,5,6)$.}
$$(3,4,1,2)\rightarrow (1,4,3,2)\rightarrow (1,2,3,4).$$
Nevertheless, it was conjectured by Kececioglu and Sankoff~\cite{KececiogluS95} that there is an optimal sorting that does not cut strips other than at their first or last elements. This conjecture was proved by Hannenhalli and Pevzner in~\cite{HannenhalliP96}. More precisely, they proved that there is an optimal sorting that does not cut strips of length at least three. 

It is common for bioinformatics applications, to consider \emph{signed permutations} (see, e.g.,~\cite{Pevzner00}). In a signed  permutation $\pis=(\pi_1,\ldots,\pi_n)$, each element $\pi_i$ has its \emph{sign} ``$-$'' or ``$+$''. Then for $i,j\in\{1,\ldots,n\}$, the reversal reverse the sign of each element $\pi_i,\ldots,\pi_j$ besides reversing their order. 
We generalize this notion and define \emph{partially signed linear permutations}, where each element has one of the sings: ``$-$'', ``$+$'' or ``no sign''. Formally, a partially signed linear permutation is 
$\pis=(\langle \pi_1,s_1\rangle,\ldots,\langle \pi_n,s_n\rangle)$ with $s_i\in\{-1,+1,0\}$ for $i\in\{1,\ldots,n\}$. Then for $1\leq i\leq j\leq n$, the reversal
$$\rhos(i,j)\circ\pis=(\langle\pi_1,s_1\rangle,\ldots,\langle\pi_{i-1},s_{i-1}\rangle,\langle\pi_j,-s_j\rangle,\ldots,\langle\pi_i,-s_i\rangle,
\langle\pi_{j+1}, s_{j+1}\rangle\ldots,\langle\pi_n,s_n\rangle).$$
We say that  $\pis=(\langle \pi_1,s_1\rangle,\ldots,\langle \pi_n,s_n\rangle)$ is \emph{signed}
if $s_i=-1$ or $s_i+1$ for each $i\in\{1,\ldots,n\}$, and $\pis$ is \emph{unsigned} if $s_i=0$ for every $i\in\{1,\ldots,n\}$.
We define the \emph{signed linear identity permutation} as
$\iotas=(\langle 1,+1\rangle,\ldots,\langle n,+1\rangle)$. 

We say that a partially signed linear permutation
$\pis=(\langle \pi_1,s_1\rangle,\ldots,\langle \pi_n,s_n\rangle)$ \emph{agrees in signs} with a signed linear permutation 
$\pis'=(\langle \pi_1,s_1'\rangle,\ldots,\langle \pi_n,s_n'\rangle)$ if $s_i=s_i'$ for $i\in\{1,\ldots,n\}$ such that $s_i\neq 0$, that is, the zero signs are replaced by either $-1$ or $+1$ in the signed permutation. For a partially signed linear permutation $\pis$, $\Sigma(\pis)$,  denotes the set of all signed linear permutations $\pis'$  that agree in signs with $\pis$.  
The \emph{reversal distance} $\ds(\pis,\sigmas)$
between a partially signed linear permutation $\pis$ and a signed linear permutation $\sigmas$ is the minimum number or reversal needed to obtain from $\pis$ a partially signed linear permutation $\pis'$ that agrees in signs with $\sigmas$;  $\ds(\pis)=\ds(\pis,\iotas)$. We say that a sequence of reversals of minimum length that result in a partially signed linear permutation 
that agrees in signs with $\iotas$ is an \emph{optimal sorting} sequence. 
It is straightforward to observe the following.

\begin{observation}\label{obs:sign-lin}
For every partially signed linear permutation 
$$\ds(\pis)=\min\{\ds(\pis')\mid {\pis}'\in \Sigma(\pis)\}.$$
\end{observation}

We generalize the results of Hannenhalli and Pevzner in~\cite{HannenhalliP96} for partially signed linear permutations. Let $\pis=(\langle \pi_1,s_1\rangle,\ldots,\langle \pi_n,s_n\rangle)$. For $1\leq i\leq j\leq n$, $(\langle \pi_i,s_i\rangle,\ldots,\langle\pi_j,s_j\rangle)$ is an \emph{interval} of $\pis$. An interval $(\langle \pi_i,s_i\rangle,\ldots,\langle\pi_j,s_j\rangle)$ is a \emph{signed block} if  $i=j$ or $i<j$  and the following holds: 
\begin{itemize}
\item[(i)] for every $h\in\{i+1,\ldots,j\}$, $|\pi_{h-1}-\pi_h|=1$,
\item[(ii)] the  block is \emph{canonically signed}, that is, $s_h\in\{0,+1\}$ if $\pi_i<\ldots<\pi_j$ and $s_h\in\{0,-1\}$ if $\pi_i>\ldots>\pi_j$.
\end{itemize}
Similarly to unsigned permutations, an inclusion maximal signed block is called a \emph{signed strip}.
A reversal $\rhos(p,q)$ \emph{cuts} a signed strip   $(\langle \pi_i,s_i\rangle,\ldots,\langle\pi_j,s_j\rangle)$ if either $i<p\leq j$ or $i\leq q<j$.

Let $\pis=(\langle \pi_1,s_1\rangle,\ldots,\langle \pi_n,s_n\rangle)$ and $\pis'=(\langle \pi_1,s_1'\rangle,\ldots,\langle \pi_n,s_n'\rangle)$ be signed linear permutations that may differ only in signs and let $\sigma=(\langle \pi_i,s_i\rangle,\ldots,\langle\pi_j,s_j\rangle)$ be a signed strip of $\pis$. It is said that $\pis$ and $\pis'$ are \emph{twins} with respect to $\sigma$ if $s_h=s_h'$ for all $h\in \{1,\ldots,i-1\}\cup\{j+1,\ldots,n\}$, that is, the signs may be only different for elements of $\sigma$. The crucial nontrivial claim of Hannenhalli and Pevzner that was used to show that  soring of unsigned permutations can be done without cutting strips of length at least three is Lemma~3.2 of~\cite{HannenhalliP96}.

\begin{lemma}[\cite{HannenhalliP96}]\label{lem:signed}
Let $\pis$ and $\pis'$ be signed linear permutations that are twins with respect to a signed strip $\sigma$ of $\pis$ with $|\sigma|\geq 3$. Then $\ds(\pis)\leq\ds(\pis')$.
\end{lemma}

Further, Hannenhalli and Pevzner used the result of Kececioglu and Sankoff~\cite{KececiogluS94} that for signed permutations, it is always possible to avoid cutting strips.

\begin{proposition}[\cite{KececiogluS94}]\label{prop:strips-cut-signed}
For a signed linear permutation $\pis$, there is an optimal sorting sequence such that no reversal cuts a signed strip.
\end{proposition}

Then the result of Hannenhalli and Pevzner~\cite{HannenhalliP96} is obtained by combining Observation~\ref{obs:sign-lin}, Lemma~\ref{lem:signed} and Proposition~\ref{prop:strips-cut-signed}.
We use the same arguments for partially signed linear permutations and the proof of the following lemma essentially repeats the proof of Theorem~3.1 of~\cite{HannenhalliP96} and we give it here for completeness.

\begin{lemma}\label{lem:strip-cut-linear}
For a partially signed linear permutation $\pis$, there is an optimal sorting sequence such that no reversal cuts a signed strip of length at least three.
\end{lemma}

\begin{proof}
Let $\pis$ be a partially signed linear permutation. The lemma is proved by the induction on $d=\ds(\pis)$. The claim is straightforward if $d\leq 1$. Assume that $d\geq 2$ and the claim holds for the lesser values. By Observation~\ref{obs:sign-lin}, there is a signed permutation $\pi'\in\Sigma(\pis)$ such that $\ds(\pis)=\ds(\pis')$. By Lemma~\ref{lem:signed}, we can assume that every signed strip $\sigma$ of $\pis$ of length at least 3 is a signed strip of $\pis'$, i.e., $\sigma$ remains canonically ordered when  zero signs in $\pis$ are replaced by $-1$ or $+1$ to construct $\pis'$. Then, by Proposition~\ref{prop:strips-cut-signed}, there is  an optimal sorting sequence for $\pis'$ such that no reversal cuts a signed strip of this permutation. Let $\rhos(i,j)$ be the first reversal in this sorting sequence. We apply it for $\pis$ and denote $\pis^*=\rhos(i,j)\circ \pis$. Note that $\rhos(i,j)$ does not cuts signed strips of $\pis$ of length at least  three. We also have that 
$\ds(\pis^*)\leq\ds(\rhos(i,j)\circ \pi')=d-1$. By induction, there is an optimal sorting sequence for $\pis^*$ such that no reversal cuts a signed strip of length at least three. This completes the proof.
\end{proof}

In our study of Whitney switches, we are interested in circular permutations and, therefore,  we extend Lemma~\ref{lem:strip-cut-linear} for such permutations. For this, we define a \emph{partially signed circular permutation}   
$\pis^c=(\langle \pi_1,s_1\rangle,\ldots,\langle \pi_n,s_n\rangle)$, where $(\pi_1,\ldots,\pi_n)$ is a linear permutation and $s_i\in\{-1,+1,0\}$ for $i\in\{1,\ldots,m\}$,
 as the class of the linear permutations that can be obtained from $(\langle \pi_1,s_1\rangle,\ldots,\langle \pi_n,s_n\rangle)$ by rotations and reflections such that every reflection reverse signs. In other words, the linear permutations 
$$(\langle \pi_1,s_1\rangle,\ldots,\langle \pi_n,s_n\rangle),(\langle \pi_n,s_n\rangle,\langle \pi_1,s_1\rangle\ldots,\langle \pi_{n-1},s_{n-1}\rangle),\ldots,(\langle \pi_2,s_2\rangle,\ldots,
\langle \pi_n,s_n\rangle,\langle \pi_1,s_1\rangle) $$
and
\begin{align*}
(\langle \pi_n,-s_n\rangle,\ldots,\langle \pi_1,-s_1\rangle),(\langle \pi_1,-s_1\rangle,\langle \pi_{n},-s_{n}\rangle\ldots,\langle \pi_{1},-s_{2}\rangle),\ldots,\\(\langle \pi_{n-1},-s_2\rangle,\ldots,
\langle \pi_1,-s_1\rangle,\langle \pi_n,-s_n\rangle)
\end{align*}
are identified.
For $i,j\in\{1,\ldots,n\}$, the reversal 
$$\rhos^c(i,j)\circ\pis^c=(\langle\pi_1,s_1\rangle,\ldots,\langle\pi_{i-1},s_{i-1}\rangle,\langle\pi_j,-s_j\rangle,\ldots,\langle\pi_i,-s_i\rangle,
\langle\pi_{j+1}, s_{j+1}\rangle\ldots,\langle\pi_n,s_n\rangle)$$
if $i\leq j$, and
$$\rhos^c(i,j)\circ\pis^c=(\langle\pi_n,-s_n\rangle,\ldots,\langle\pi_{i},-s_{i}\rangle,\langle\pi_{j+1},s_{j+1}\rangle,\ldots,\langle\pi_{i-1},s_{i-1}\rangle,
\langle\pi_{j}, -s_{j}\rangle\ldots,\langle\pi_1,-s_1\rangle)$$
otherwise.

In the same way as with partially signed linear permutations, $\pis^c$ is \emph{signed} if each $s_i$ is either $-1$ or $+1$ and the \emph{signed circular identity permutation} is
$\iotas^c=(\langle 1,+1\rangle,\ldots,\langle n,+1\rangle)$. 
Also a partially signed circular permutation
$\pis^c=(\langle \pi_1,s_1\rangle,\ldots,\langle \pi_n,s_n\rangle)$ \emph{agrees in signs} with a signed circular permutation 
$\pis'^c=(\langle \pi_1,s_1'\rangle,\ldots,\langle \pi_n,s_n'\rangle)$ if $s_i=s_i'$ for $i\in\{1,\ldots,n\}$ such that $s_i\neq 0$, that is, the zero signs are replaced by either $-1$ or $+1$ in the signed permutation, and  $\Sigma(\pis^c)$ is used to denote the set of all signed circular  permutations $\pis'^c$  that agree in signs with $\pis^c$.
Then  \emph{reversal distance} $\ds^c(\pis^c,\sigma^c)$, where $\sigmas^c$ is a signed circular permutation, is the minimum number or reversal needed to obtain from $\pis^c$ a  partially signed circular permutation $\pis'^c$ that agrees in signs with $\sigmas^c$, and $\ds^c(\pis^c)=\ds^c(\pis^c,\iotas^c)$. A a sequence of reversals of minimum length that result in a partially signed circular permutation 
that agrees in signs with $\iotas^c$ is an \emph{optimal sorting} sequence. 

We exploit the following properties of partially signed  permutations. To state them, we need some auxiliary notation. 
For a partially signed linear permutation $\pis=(\langle \pi_1,s_1\rangle,\ldots,\langle \pi_n,s_n\rangle)$, we define the negation 
$-\pis=(\langle \pi_n,-s_n\rangle,\ldots,\langle \pi_1,-s_1\rangle)$.
For an integer $h$, we denote $\pis\oplus h=(\langle \pi_{1+h},s_{1+h}\rangle,\ldots,\langle \pi_{n+h},s_{n+h}\rangle)$, where it is assumed that $\pi_0=\pi_n$, $s_0=s_n$ and the other indices are taken modulo $n$. The negation corresponds to the reflection and $\oplus$ defines rotations.

\begin{lemma}\label{lem:rotate}
Let 
$\pis$  be partially signed linear permutation, $\sigmas$ be a signed permutation,
and let $h$ be an integer.
Then
$$\min\{\ds(\pis,\sigmas),\ds(\pis,-\sigmas)\}=\min\{\ds(\pis\oplus h,\sigmas\oplus h),\ds(\pis\oplus h,-(\sigmas\oplus h)\}.$$
\end{lemma}

\begin{proof}
We show that $$\min\{\ds(\pis,\sigmas),\ds(\pis,-\sigmas)\}\geq\min\{\ds(\pis\oplus h,\sigmas\oplus h),\ds(\pis\oplus h,-(\sigmas\oplus h)\}.$$
The proof of the opposite inequality is symmetric and is done by replacing $h$ by $-h$.

The proof is by induction on the distance between permutations. Let 
$\pis=(\langle \pi_1,s_1\rangle,\ldots,\langle \pi_n,s_n\rangle)$ and $\sigmas$ be arbitrary partially signed and signed linear permutations, respectively, with 
 $d=\min\{\ds(\pis,\sigmas),\ds(\pis,-\sigmas)\}$. The claim is trivial for $d=0$. Let $d\geq 1$ and assume that the claim holds for every two permutations at reversal distance at most $d-1$. We assume without loss for generality that $d=\ds(\pis,\sigmas)$, as the other case is symmetric. 

Consider the corresponding sequence of reversals of length $d$ and assume that $\rhos(i,j)$ for $1\leq i\leq j\leq n$ is the first reversal in the sequence. 
Recall that 
$$\pis'=\rhos(i,j)\circ\pis=(\langle\pi_1,s_1\rangle,\ldots,\langle\pi_{i-1},s_{i-1}\rangle,\langle\pi_j,-s_j\rangle,\ldots,\langle\pi_i,-s_i\rangle,
\langle\pi_{j+1}, s_{j+1}\rangle\ldots,\langle\pi_n,s_n\rangle).$$
Note that either $i\neq 1$ or $j\neq n$, because $d\geq\ds(\pis,-\sigmas)$.
Let $i'=(i+h)\mod n$ and $j'=(j+h)\mod n$ assuming that $n\mod n=n$. Let $\pis^*=\pis\oplus h$.

Suppose that $i'\leq j'$. Then $\rhos(i',j')\circ \pis^*=\pis'\oplus h$. By the inductive assumption,
\begin{align*}
d-1\geq &\min\{\ds(\pis',\sigmas),\ds(\pis',-\sigmas)\}\geq\min\{\ds(\pis'\oplus h,\sigmas\oplus h),\ds(\pis'\oplus h,-(\sigmas\oplus h))\}\\
\geq &\min\{\ds(\pis\oplus h,\sigmas\oplus h),\ds(\pis\oplus h,-(\sigmas\oplus h))\}-1,
\end{align*}
and the claim follows.

Assume that $i'>j'$. Let $i''=j'+1$ and $j''=i'-1$. Notice that since $i\neq 1$ or $j\neq n$, $i''\leq j''$. 
Then $\rhos(i'',j'')\circ \pis^*=-(\pis'\oplus h)$. Using the inductive assumption we obtain that 
\begin{align*}
d-1\geq &\min\{\ds(\pis',\sigmas),\ds(\pis',-\sigmas)\}\geq\min\{\ds(\pis'\oplus h,\sigmas\oplus h),\ds(\pis'\oplus h,-(\sigmas\oplus h))\}\\
\geq &\min\{\ds(-(\pis'\oplus h),\sigmas\oplus h),\ds(-(\pis'\oplus h),-(\sigmas\oplus h))\}\\
\geq &\min\{\ds(\pis\oplus h,\sigmas\oplus h),\ds(\pis\oplus h,-(\sigmas\oplus h))\}-1.
\end{align*}
This completes the proof.
\end{proof}

\begin{lemma}\label{lem:linearize}
Let $\pis$ be a partially signed circular permutation. 
Then 
$$\ds^c(\pis^c)=\min\{\ds(\sigmas)\mid \sigmas\in\pis^c\}.$$
\end{lemma}

\begin{proof}
Clearly, for every $\sigmas\in\pis^c$, $\ds^c(\pis^c)\leq \ds(\sigmas)$. Therefore, we have to show that there is $\sigmas\in\pis^c$ such that $\ds^c(\pis^c)\geq \ds(\sigmas)$. 
Let $\pis^c=(\langle \pi_1,s_1\rangle,\ldots,\langle \pi_n,s_n\rangle)$ and let $d=\ds^c(\pis^c)$.

We claim that there is an integer $h$ such that for the partially signed linear permutation $\pis=(\langle \pi_1,s_1\rangle,\ldots,\langle \pi_n,s_n\rangle)$,
$\min\{\ds(\pis,\iotas\oplus h),\ds(\pis,-(\iotas\oplus h))\}\leq d$.

The proof is by the induction on $d$. The claim is trivial if $d=0$. Let $d\geq 1$ and assume that the claim holds for all partially signed circular permutations $\pis'$ with $\ds^c(\pis')\leq d-1$. 
Consider an optimal sorting sequence for $\pis^c$ and let $\rho^c(i,j)$ be the first reversal in the sequence. Let $\pis'^c=\rhos^c(i,j)\circ\pis^c$.

Suppose that $i\leq j$. Let  $\pis'=\rhos(i,j)\circ\pis$. By the inductive assumption, there is $h$ such that 
$\min\{\ds(\pis',\iotas\oplus h),\ds(\pis',-(\iotas\oplus h))\}\leq d-1$. Therefore, $\min\{\ds(\pis,\iotas\oplus h),\ds(\pis,-(\iotas\oplus h))\}\leq d$.

Let $i>j$. If $(j+1)-i=0 \mod n$, that is, the indices $j$ and $i$ are consecutive in the cycle ordering, then $\rhos^c(i,j)$ just reflects $\pis^c$ contradicting the optimality of the chosen sorting sequence. Hence, for $i'=j+1$ and $j'=i-1$, we have that $i'\leq j'$. Let $\pis'=\rhos(i',j')\circ\pis$.
By induction, there is an integer $h$ such that 
$\min\{\ds(-\pis',\iotas\oplus h),\ds(-\pis',-(\iotas\oplus h))\}\leq d-1$. Clearly, 
$\min\{\ds(-\pis',\iotas\oplus h),\ds(-\pis',-(\iotas\oplus h))\}=\min\{\ds(\pis',\iotas\oplus h),\ds(\pis',-(\iotas\oplus h))\}$ and, therefore, 
$\min\{\ds(\pis,\iotas\oplus h),\ds(\pis,-(\iotas\oplus h))\}\leq d$. This competes the proof of the auxiliary claim.

To prove the lemma, observe that by Lemma~\ref{lem:rotate}, we obtain that 
$$
\min\{\ds(\pis\oplus(-h),\iotas),\ds(\pis\oplus(-h),-\iotas)\}=\min\{\ds(\pis,\iotas\oplus h),\ds(\pis,-(\iotas\oplus h))\}\leq d
$$
If $\ds(\pis\oplus(-h),\iotas)\leq\ds(\pis\oplus(-h),-\iotas)$, we set $\sigmas=\pis\oplus(-h)$ and  $\sigmas=-(\pis\oplus(-h))$ otherwise. It is straightforward to see that $\sigmas\in\pis^c$ and this completes the proof.
\end{proof}

The notion of signed strips can be extended for partially signed circular permutations in a natural way. More formally, this is done as follows.
 Let $\pis^c=(\langle \pi_1,s_1\rangle,\ldots,\langle \pi_n,s_n\rangle)$ be a partially signed circular permutation. 
For $1\leq i\leq j\leq n$, we say that $(\langle \pi_i,s_i\rangle,\ldots,\langle\pi_j,s_j\rangle)$ and 
$(\langle \pi_{j+1},s_{j+1}\rangle,\ldots,\langle\pi_n,s_n\rangle,\langle \pi_{1},s_{1}\rangle,\ldots,\langle\pi_i,s_i\rangle)$
are \emph{intervals} of $\pis^c$. 
An interval is a \emph{signed block} if it either has size one or  for every two consecutive elements $\langle \pi_{i-1},s_{i-1}\rangle,\langle \pi_i,s_i\rangle$, 
$|\pi_{i-1}-\pi_i|\leq 1$ and, moreover, if the elements of the interval are in the increasing order, then all the signs $s_i\in\{0,+1\}$, and if they are in the 
the decreasing order, then all the signs $s_i\in\{0,-1\}$. 
A \emph{signed strip} is an inclusion maximal signed block. 
A reversal $\rhos^c(p,q)$ \emph{cuts} an interval if the reversed part includes at least one element of the interval and excludes at least one element of the interval.

We conjecture that the result of Hannenhalli and Pevzner~\cite{HannenhalliP96} can be extended  for partially signed circular permutations in the same way as for the linear case in Lemma~\ref{lem:strip-cut-linear}. However, it seems that for this, the variant of Lemma~\ref{lem:signed} for circular permutations should be proved. This can be done by following and adjusting the  arguments from~\cite{HannenhalliP96}. The proof of Lemma~\ref{lem:signed} is nontrivial and is based on the deep duality theorem of Hannenhalli and Pevzner~\cite{HannenhalliP99} that is also is sated for linear permutations. Hence, proving the circular analog of Lemma~\ref{lem:signed} would demand a lot of technical work and this goes beyond of the scope of our paper. Therefore, we show the simplified claim that can be derived from Lemma~\ref{lem:strip-cut-linear} and is sufficient for our purposes.

\begin{lemma}\label{lem:strip-cut-circ}
For a signed circular permutation $\pis^c$, there is an optimal sorting sequence such that no reversal in the sequence cuts the interval formed by a signed strip of $\pis^c$ of length at least $5$.
\end{lemma}

Notice that we do not claim that no reversal cuts a strip of length at least 5 that is obtained by performing the previous reversals; only the long strips of the initial permutation $\pis^c$ are not cut by any reversal in the sorting sequence.  

\begin{proof}
Let $\pis^c$ be a partially signed circular permutation.  By Lemma~\ref{lem:linearize}, there is a partially signed linear permutation $\sigmas\in\pis^c$ such that 
$d=\ds(\sigma)=\ds^c(\pis^c)$. Let 
$\sigmas=(\langle \sigma_1,s_1\rangle,\ldots,\langle \sigma_n,s_n\rangle)$. Note that, by definition, we can write that
$\pis^c=(\langle \sigma_1,s_1\rangle,\ldots,\langle \sigma_n,s_n\rangle)$. We assume that $\ds^c(\pis^c)\geq 1$. We consider 3 cases.

\medskip
\noindent
{\bf Case~1.} Every signed strip of length at least 5 of $\pis^c$ is a signed strip of the linear permutation $\sigmas$. Consider an optimal sorting sequence for $\sigmas$ that does not cut strips of length at least 5 that exists by Lemma~\ref{lem:strip-cut-linear}. Clearly, this sequence is an optimal sorting sequence for $\pis^c$ satisfying the claim.

\medskip
\noindent
{\bf Case~2.}
There is a unique signed strip $\omega=(\langle \sigma_i,s_i\rangle,\ldots,\langle \sigma_j,s_j\rangle )$ for $i\leq i<j\leq n$ 
of $\pis^c$ with length at least 5 that is not a signed strip of $\sigma$. Then 
\begin{equation}\label{eq:omega}
\omega=(\langle p,s_i\rangle,\ldots,\langle n,s_{n-p}\rangle,\langle 1,s_{n-p+1}\rangle,\ldots,\langle ((p+j-i)\mod n),s_j\rangle)
\end{equation}
 for some  $p\geq n-(j-i)+1$ 
or, symmetrically, 
$$\omega=(\langle p, s_i\rangle,\ldots,\langle 1,s_{n-p}\rangle,\langle n,s_{n-p+1}\rangle,\ldots,\langle n-(j-i)+p,s_j\rangle).$$
Using symmetry, we assume without loss of generality that $\omega$ is of form (\ref{eq:omega}) and 
 write that $\omega=\omega'\omega''$, where 
$\omega'= (\langle p,s_i\rangle,\ldots,\langle n,s_{n-p}\rangle)$ and $\omega''=(\langle 1,s_{n-p+1}\rangle\ldots,\langle ((p+j-i)\mod n),s_j\rangle)$
Since $j-i\geq 4$, either $|\omega'|\geq 3$ or $|\omega''|\geq 3$. 
Assume that $|\omega'|\geq 3$ as the other case is completely symmetric. 
By Lemma~\ref{lem:strip-cut-linear}, there is an optimal sorting sequence $\mathcal{S}$ for $\sigmas$ that does not cut strips of length at least 3. In particular, $\omega'$ is not cut by any reversal in the sequence. We modify $\mathcal{S}$ as follows for every reversal:
\begin{itemize}
\item exclude the elements of $\omega''$ from the reversed interval and its complement,
\item if  the reversed interval includes either $w'$ or $-w'$, then replace $w'$ by $w$,
\item if the complement of the reversed interval contains either $w'$ or $-w'$, then replace $w'$ by $w$.
\end{itemize}
In other words, whenever we reverse $w'$, we reverse it together with $w''$, and if we do not reverse $w'$, we keep $w''$ together with $w'$ and do not reverse the elements of $w''$. Let $\mathcal{S}'$ be the obtained sequence. It is straightforward to verify that every step of $\mathcal{S}'$ is indeed a reversal and no reversal cuts a strip of $\pis^c$ of length at least 5. Moreover, 
after performing all the reversals of $\mathcal{S}'$ we obtain the partially signed permutation that agrees in signs with $\iotas\oplus ((p+j-i)\mod n)$ that is in $\iotas^c$. This means that $\mathcal{S}'$ is a sorting sequence for $\pis^c$ of length  $d$.

\medskip
\noindent
{\bf Case~3.}  Three are $1\leq i<j\leq n$ such that 
$$\omega=(\langle \sigma_j, s_j\rangle,\ldots,\langle \sigma_n,s_n\rangle,\langle \sigma_1,s_1\rangle,\ldots,\langle \sigma_i,s_i\rangle)$$
  is a strip of $\pis^c$. Let $\sigmas'=\sigma\oplus(-i)=(\langle \sigma_1',s_1'\rangle,\ldots,\langle \sigma_n',s_n'\rangle)$. By Lemma~\ref{lem:rotate}, 
$$d=\min\{\ds(\sigmas',\iotas\oplus (-i)),\ds(\sigmas,-(\iotas\oplus (-i))\}.$$ Since the cases are symmetric, assume without loss of generality that 
$\ds(\sigmas',\iotas\oplus (-i))=d$. Consider $\sigmas''=\sigma\oplus(-i)=(\langle \sigma_1'',s_1'\rangle,\ldots,\langle \sigma_n'',s_n'\rangle)$, where 
$\sigma_i''=(\sigma_i'+i) \mod n$ (assuming that $n\mod n=n$). We have that $\ds(\sigmas'')=d$ and sorting of $\sigmas''$ is equivalent to computing the minimum sequence of reversals needed to transform $\sigma'$ to a partially signed permutation that agrees in signs with $\iotas\oplus (-i)$. Note that sorting of the circular partially signed permutation 
$\sigmas''^c$ is equivalent to sorting $\rhos^c$ and $\sigmas''^c$ has no strips including $\langle \sigma'',s_n'\rangle$ and $\langle \sigma_1'',s_1\rangle$.
 Finally, observe that either  every signed strip of length at least 5 of $\sigmas''^c$ is a signed strip of the linear permutation $\sigmas''$ and we are in Case~1 or
there is a unique signed strip $\omega''=(\langle \sigma_i'',s_i'\rangle,\ldots,\langle \sigma_j'',s_j''\rangle )$ for $i\leq i<j\leq n$ 
of $\sigmas''^c$ with length at least 5 that is not a signed strip of $\sigma''$ and we are in Case~2.
\end{proof}

We conclude the section by observing that if the elements of a partially signed circular permutation are ordered, then the sorting can be done easily. We say that the reversal $\rhos^c(i,j)$ is \emph{trivial} %if ether 
$i=j$. 

\begin{lemma}\label{lem:signs}
Let $\pis^c=(\langle 1,s_1\rangle,\ldots,\langle n,s_n\rangle)$. Then $\ds^c(\pis^c)=|I|$, where $I=\{i\mid 1\leq i\leq n,s_i=-1\}$ and the reversals $\rhos(i,i)$ for $i\in I$ compose an optimal sorting sequence.
\end{lemma}

\begin{proof}
Let $i\in I$. Let $\mathcal{S}$ be an optimal sorting sequence. We assume that $\mathcal{S}$ does not contain reversals $\rhos^c(i-2,i)$ for $i\in\{1,\ldots,n\}$
(as before, we take the values modulo $n$ assuming that $n\mod n=n$), because they are equivalent to $\rhos^c(i,i)$. 
Observe that if $i\in I$, then the intervals $(\langle (i-1),s_{i-1}\rangle,\langle i,s_i\rangle)$ and $(\langle i,s_i\rangle,\langle (i+1),s_{i+1}\rangle)$ should be split by some reversals from $\mathcal{S}$. Moreover, we can observe the following for  $i-1,i\in I$. Assume that $\rhos^c(p,q)$ is the first reversal that splits $(\langle (i-1),s_{i-1}\rangle,\langle i,s_i\rangle)$ and assume that $i-1$ keeps its sign $s_{i-1}$. Let $\sigma=(\langle (i-1),s_{i-1}\rangle,\langle j,s_j\rangle)$ be the interval composed by 
$(\langle (i-1),s_{i-1}\rangle$ and the next element after applying $\rhos^c(p,q)$. If the reversal is trivial, then $\sigma=(\langle (i-1),s_{i-1}\rangle,\langle i,-s_i\rangle)$ and $\sigma$ should be split again. If $j\neq i$, then we have to split $\sigma$, because $j\neq i-2$ and, therefore, $|j-(i-1)|>1$. These observations imply that $\mathcal{S}$ contains at least $|I|$ reversals.
Therefore, the sorting sequence formed by the reversals $\rhos(i,i)$ for $i\in I$ is optimal.
\end{proof}

\section{Tutte decomposition and 2-isomorphisms}\label{sec:techn}
 In this section we provide a number of auxiliary results about 2-isomorphisms and Tutte decompositions.

Recall that for two $n$-vertex 2-connected graphs $G$ and $H$, a bijective mapping $\varphi\colon E(G)\rightarrow E(H)$ is a 2-isomorphism if  $\varphi$ and $\varphi^{-1}$ preserve cycles.  
We also say that an isomorphism $\psi\colon V(G)\rightarrow V(H)$ is a $\varphi$-isomorphism if for every edge $uv\in E(G)$, $\varphi(uv)=\psi(u)\psi(v)$, and $G$ and $H$ are $\varphi$-isomorphic if there is an isomorphism $G$ to $H$ that is a $\varphi$-isomorphism. We need the following folklore observation about $\varphi$-isomorphisms that we prove for completeness.
For this, we extend $\varphi$ on sets of edges in standard way, that is, $\varphi(A)=\{\varphi(e)\mid e\in A\}$ and $\varphi(\emptyset)=\emptyset$. 

\begin{lemma}\label{lem:folk} 
Let $G$ and $H$ be $n$-vertex $2$-connected $2$-isomorphic graphs with a $2$-isomorphism $\varphi$. Then $G$ and $H$ are $\varphi$-isomorphic if and only if 
there is a bijective mapping $\psi\colon V(G)\rightarrow V(H)$ such that 
for every $v\in V(G)$, $\varphi(E_G(v))=E_{H}(\psi(v))$. Moreover, $G$ and $H$ are  $\varphi$-isomorphic if and only if $\varphi$ bijectively maps the family of the sets of edges
$\{E_G(v)\mid v\in V(G)\}$ to the family $\{E_H(v)\mid v\in V(H)\}$, and this property can be checked in polynomial time.
\end{lemma}

\begin{proof}
If $\psi$ is an $\varphi$-isomorphism of $G$ to $H$, then $\varphi(E_G(v))=E_{H}(\psi(v))$ for all $v\in V(G)$ by the definition. For the opposite direction, assume that $\psi\colon V(G)\rightarrow V(H)$ is a bijection such that $\varphi(E_G(v))=E_{H}(\psi(v))$ for every $v\in V(G)$.  Suppose that $u$ and $v$ are distinct vertices of $G$. We claim that $u$ and $v$ are adjacent in $G$ if and only if $\psi(u)$ and $\psi(v)$ are adjacent in $H$. 
Suppose that $u$ and $v$ are adjacent in $G$. Then $E_G(u)\cap E_G(v)=\{uv\}$. Therefore,  
$E_H(\psi(u))\cap E_H(\psi(v))=\varphi(E_G(u))\cap \varphi(E_G(v))=\{\varphi(uv)\}$. This means that $\psi(u)$ and $\psi(v)$ are adjacent in $H$ and $\psi(u)\psi(v)=\varphi(uv)$. If $u$ and $v$ are not adjacent, then  $E_G(u)\cap E_G(v)=\emptyset$ and  $E_H(\psi(u))\cap E_H(\psi(v))=\varphi(E_G(u))\cap \varphi(E_G(v))=\emptyset$, that is, $\psi(u)$ and $\psi(v)$ are not adjacent in $H$.

The second claim of the lemma immediately follows from the first. 
\end{proof}

By Lemma~\ref{lem:folk}, we can restate the task of \probWS and ask whether it is possible to obtain a graph $G'$ by performing at most $k$ Whitney switches starting from $G$ with the property 
 that the extension of $\varphi$ to  the family of sets $\{E_{G'}(v)\mid v\in V(G')\}$ bijectively maps this family to $\{E_H(v)\mid v\in V(H)\}$.
 
We use Whitney's theorem~\cite{Whitney33}(see also~\cite{Truemper80}).

\begin{proposition}[\cite{Whitney33}]\label{prop:whitney}
Let $G$ and $H$ be $n$-vertex graphs and let $\varphi$ be a $2$-isomorphism of $G$ to $H$. Then  there is a finite sequence of Whitney switches such that the graph $G'$ obtained from $G$ by these switches is $\varphi$-isomorphic to $H$.
\end{proposition}

We also use the property of 3-connected graphs explicitly given by Truemper~\cite{Truemper80}. It also can be derived  from Proposition~\ref{prop:whitney}.

\begin{proposition}[\cite{Truemper80}]\label{prop:three-conn}
Let $G$ and $H$ be 3-connected $n$-vertex graphs and let $\varphi$ be a $2$-isomorphism of $G$ to $H$. Then $G$ and $H$ are $\varphi$-isomorphic. 
\end{proposition}

Throughout this section we assume that  $G$ and $H$ are $n$-vertex 2-connected graphs and let $\varphi$ be a $2$-isomorphism of $G$ to $H$. Let also $\mathcal{T}^{(1)}=(T^{(1)},\{X_t^{(1)}\}_{t\in V(T^{(1)})})$
and  $\mathcal{T}^{(2)}=(T^{(2)},\{X_t^{(2)}\}_{t\in V(T^{(2)})})$ be the Tutte decompositions of $G$ and $H$, respectively, and denote by $(W_2^{(h)},W_{\geq 3}^{(h)})$ the partition of $V(T^{(h)})$ satisfying (T4)--(T8) for $h=1,2$.

The following lemma is crucial for us.

\begin{lemma}\label{lem:isom-decomp}
There is  an isomorphism $\alpha$ of $T^{(1)}$ to $T^{(2)}$ such that  
\begin{itemize}
\item[(i)] for every $t\in V(T^{(1)})$, $|X_t^{(1)}|=|X_{\alpha(t)}^{(2)}|$, in particular, $t\in W_2^{(1)}$ ($t\in W_{\geq 3}^{(1)}$, respectively)  if and only if $\alpha(t)\in W_2^{(2)}$  ($\alpha(t)\in W_{\geq 3}^{(2)}$, respectively),
\item[(ii)] for every $t\in W_{\geq 3}^{(1)}$, the torso of $X_t^{(1)}$ is a 3-connected graph (a cycle, respectively) if and only if the torso of $X_{\alpha(t)}^{(2)}$ is a 3-connected graph (a cycle, respectively),
\item[(iii)] for every $t\in V(T^{(1)})$, 
$\varphi(E(G[X_t^{(1)}])=E(H[X_{\alpha(t)}^{(2)}])$.
\end{itemize}
\end{lemma}

\begin{proof}
By Proposition~\ref{prop:whitney}, there is a finite sequence of Whitney switches such that the graph $G'$ obtained from $G$ by these switches is $\varphi$-isomorphic to $H$. We prove the lemma by induction on the number of switches. The claim is straightforward if this number is zero, because $G$ and $H$ are $\varphi$-isomorphic. 
It is sufficient to observe that the Tutte decomposition is unique by Proposition~\ref{prop:tutte} and then use Lemma~\ref{lem:folk}. Assume that the sequence has length $\ell\geq 1$ and  the claim of the lemma holds  for the sequences of length at most $\ell-1$. 

Let $(A,B)$ be a Whitney separation of $G$ such that the first switch is done with respect to $(A,B)$. Denote by $\{u,v\}=A\cap B$.
Denote by $G'$ the graph obtained from $G$ by the Whitney switch with respect to $(A,B)$. Recall that $G'$ is constructed by replacing each edge $ux\in E(G)$ for $x\in B\setminus A$ by $vx$ and by replacing each edge $vx$ for $x\in B\setminus A$ by $ux$.  Recall also that we denote by $\sigma_{(A,B)}$ the mapping of the edges of $G$ to the edges of the graph $G'$ obtained from $G$ by the Whitney switch with respect to $(A,B)$ that corresponds to the switch. We have that  $\sigma_{(A,B)}$ is a 2-isomorphism of $G$ to $G'$ and, by our convention, the set of edges remains the same and we only modify incidences of some of them, that is, $\sigma_{(A,B)}$ is the identity mapping.
Since $H$ is obtained from $G'$ by $\ell-1$ switches, it is sufficient to show the claim for $G'$. We do it by constructing the Tutte decomposition of $G'$ from the decomposition of $G$.

Suppose first that $A\cap B=X_t^{(1)}$ for some $t\in W_2^{(1)}$. By the definition of the Tutte decomposition, for each $s\in V(T)$, either $X_{s}^{(1)}\subseteq A$ or  $X_{s}^{(1)}\subseteq B$. We construct the tree decomposition $\mathcal{T}'=(T^{(1)},\{X_{s}'\}_{s\in V(T^{(1)})})$. For every $s\in V(T^{(1)})$ such that $X_s^{(1)}\subseteq A$, we define 
$X_s'=X_s^{(1)}$. Similarly, 
if $X_s^{(1)}\subseteq B$ and $u,v\notin X_s^{(1)}$, $X_s'=X_s^{(1)}$. For all  $s\in V(T^(1))$ such that  $X_s^{(1)}\subseteq B$, $s\neq t$ and $\{u,v\}\cap X_s^{(1)}\neq\emptyset$, we construct $X_s'$ from $X_s^{(1)}$ as follows:
\begin{itemize}
\item[a)] replace $u$ by $v$ if $u\in X_s^{(1)}$ and $v\notin X_s^{(1)}$,
\item[b)] replace  $v$ by $u$ if $v\in X_s^{(1)}$ and $u\notin X_s^{(1)}$.
\end{itemize}
It is straightforward to verify that $\mathcal{T}'$ is the Tutte decomposition and $\alpha$ that maps the nodes of $T^{(1)}$ to themselves satisfies (i)--(iii).

Assume now that $A\cap B\neq X_t^{(1)}$ for all $t\in W_2^{(1)}$. By the definition of the Tutte decomposition, this means that $u,v\in X_t^{(1)}$ for some $t\in W_{\geq 3}^{(1)}$ such that the torso of $X_t^{(1)}$ is a cycle $C$ and $u,v$ are nonadjacent vertices of $C$.  Notice that $Z_A=X_t^{(1)}\cap A$ and $Z_B=X_t^{(1)}\cap B$ induce distinct $(u,v)$-paths in $C$. 
We again construct the tree decomposition $\mathcal{T}'=(T^{(1)},\{X_{s}'\}_{s\in V(T^{(1)})})$. Notice that for each $s\in V(T^{(1)})$ such that $s\neq t$, either $X_{s}^{(1)}\subseteq A$ or  $X_{s}^{(1)}\subseteq B$. For all such $s$, we define $X_s'$ in exactly the same way as in the previous case. We define $X_t'=X_t^{(1)}$. 
It is straightforward to verify that $\mathcal{T}'$ is a tree decomposition of $G'$. 
Since the torso of $X_t^{(1)}$ is a cycle composed by the paths with the vertices $Z_A$ and $Z_B$, the torso of $X_s'$ in $\mathcal{T}'$ is a cycle as well. This implies that  
$\mathcal{T}'$ is the Tutte decomposition of $G'$. Then $\alpha$ that maps the nodes of $T^{(1)}$ to themselves satisfies (i)--(iii).
\end{proof}

Let $F$ be a 2-connected graph. Let also $\mathcal{T}=(T,\{X_t\}_{t\in V(T)})$ be the Tutte decomposition of $F$ and let $(W_2,W_{\geq 3})$ be the partition of $V(T)$ satisfying (T4)--(T8). We denote by $\widehat{F}$ the graph obtained from $F$ by making the vertices of $X_t$ adjacent for every $t\in W_2$. We say that $\widehat{F}$ is the \emph{enhancement} of $F$. Note that $\mathcal{T}$ is the Tutte decomposition of $\widehat{F}$ and the torso of each bag $X_t$ is $\widehat{F}[X_t]$.  Notice also that $(A,B)$ is a Whitney separation of $F$ if and only if $(A,B)$ is a  Whitney separation of $\widehat{F}$. We also say that $F$ is \emph{enhanced} if $F=\widehat{F}$.

To simplify the arguments in our proofs, it is convenient for us to switch from 2-isomorphisms of graphs to 2-isomorphisms of their enhancements. 
 By Lemma~\ref{lem:isom-decomp}, there is an isomorphism $\alpha$ of $T^{(1)}$ to $T^{(2)}$ satisfying conditions (i)---(ii) of the lemma. We define the \emph{enhanced} mapping $\widehat{\varphi}\colon E(\widehat{G})\rightarrow E(\widehat{H})$ such that $\widehat{\varphi}(e)=\varphi(e)$ for $e\in E(G)$, and for each $e\in E(\widehat{G})\setminus E(G)$ with its end-vertices in $X_t^{(1)}$ for some $t\in W_2^{(1)}$, we define $\widehat{\varphi}(e)$ be the edge with the end-vertices in $X_{\alpha(t)}^{(2)}$.   

\begin{lemma}\label{lem:enhanced}
The mapping
$\widehat{\varphi}$ is a $2$-isomorphism of $\widehat{G}$ to $\widehat{H}$. Moreover, a sequence of Whitney switches makes $G$ $\varphi$-isomorphic to $H$ if and only if the same sequence makes $\widehat{G}$ $\widehat{\varphi}$-isomorphic to $\widehat{H}$. 
\end{lemma}

\begin{proof}
Note that by Lemma~\ref{lem:isom-decomp}, $\widehat{\varphi}$ is a bijection. It is sufficient to show the second claim of the lemma, because if a sequence of Whitney switches makes $G$ $\varphi$-isomorphic to $H$, then $\varphi$ is a 2-isomorphism of $G$ to $H$. 
 Notice that $\widehat{G}$ and $\widehat{H}$ have the same separators of size 2 as $G$ and $H$, respectively, by the definition of the Tutte decomposition. Therefore, given a sequence of Whitney switches of $G$, the same sequence can be performed on $\widehat{G}$. Then Lemmas~\ref{lem:folk} and \ref{lem:isom-decomp} imply that if 
 a sequence of Whitney switches makes $G$ $\varphi$-isomorphic to $H$, then the same sequence makes $\widehat{G}$ $\widehat{\varphi}$-isomorphic to $\widehat{H}$.
Since it is straightforward to see that if  a sequence of Whitney switches makes $\widehat{G}$ $\widehat{\varphi}$-isomorphic to $\widehat{H}$, then the same sequence makes $G$ $\varphi$-isomorphic to $H$, the second claim holds.
\end{proof}

Lemma~\ref{lem:enhanced} allows us to consider enhanced graph and this is useful, because we can strengthen the claim of Lemma~\ref{lem:isom-decomp}.

\begin{lemma}\label{lem:isom-decomp-enh}
Let $G$ and $H$ be enhanced graphs.
Then there is  an isomorphism $\alpha$ of $T^{(1)}$ to $T^{(2)}$ such that conditions (i)--(iii) of Lemma~\ref{lem:isom-decomp} are fulfilled and, moreover, 
\begin{itemize}
\item[(iv)]  for every $t\in V(T^{(1)})$, $G[X_t^{(1)}]$ is isomorphic to $H[X_{\alpha(t)}^{(2)}]$.
\end{itemize}
Moreover, if $G[X_t^{(1)}]$ is 3-connected, then $G[X_t^{(1)}]$ is $\varphi$-isomorphic to $H[X_{\alpha(t)}^{(2)}]$.
\end{lemma}

\begin{proof}
We have that $G[X_t^{(1)}]$ is isomorphic to $H[X_{\alpha(t)}^{(2)}]$ for $t\in W_2^{(1)}$, because $G$ and $H$ are enhanced graphs.
Let $t\in W_{\geq 3}^{(1)}$. By conditions (i) and (ii) of Lemma~\ref{lem:isom-decomp}, $|X_t^{(1)}|=|X_{\alpha(t)}^{(2)}|$ and  $G[X_t^{(1)}]$ is a 3-connected graph (a cycle, respectively) if and only if $H[X_{\alpha(t)}^{(2)}]$ is a 3-connected graph (a cycle, respectively). If  $G[X_t^{(1)}]$ and $H[X_{\alpha(t)}^{(2)}]$ are cycles, then they are isomorphic. Assume that $G[X_t^{(1)}]$ and $H[X_{\alpha(t)}^{(2)}]$ are 3-connected. By (iii), $\varphi(E(G[X_t^{(1)}])=E(H[X_{\alpha(t)}^{(2)}])$. This implies that $\varphi$ is a 2-isomorphism of $G[X_t^{(1)}]$ to $H[X_{\alpha(t)}^{(2)}]$. By Proposition~\ref{prop:three-conn},  $G[X_t^{(1)}]$ are $H[X_{\alpha(t)}^{(2)}]$ isomorphic and, moreover, $\varphi$-isomorphic.
\end{proof}

For the remaining part of the sections, we assume that $G$ and $H$ are enhanced graphs and  $\alpha$ is the isomorphism of $T^{(1)}$ to $T^{(2)}$ satisfying conditions (i)--(iv) of Lemmas~\ref{lem:isom-decomp} and \ref{lem:isom-decomp-enh}. 

Our next aim is to investigate properties of the sequences of  Whitney switches that are used  in solutions for \probWS. 
For  a sequence $\mathcal{S}$ of Whitney switches such that the graph $G'$ obtained from $G$ by applying this sequence is $\phi$-isomorphic to $H$, we say that $\mathcal{S}$ is an \emph{$H$-sequence}. We also say that $\mathcal{S}$ is \emph{minimum} if $\mathcal{S}$ has minimum length. 

Recall our assumption that Whitney switches do not change the sets of vertices and edges but modify incidences of some edges. Using this assumption, we observe that we can rearrange sequences of switches in a certain way.

\begin{lemma}\label{lem:rearrange}
Let $\mathcal{S}$ be an $H$-sequence of Whitney switches and let $\mathcal{S}'$ is the sequence that differs from  $\mathcal{S}$ by the order of switches such that the following holds:
\begin{itemize}
\item[(i)] for every $t\in W_{2}^{(1)}$, the order of the switches  with respect to a Whitney separations $(A,B)$  with $A\cap B=X_t^{(1)}$ is the same as in $\mathcal{S}$,
\item[(ii)] for every $t\in W_{\geq 3}^{(1)}$,
the order of the switches with respect to Whitney separations $(A,B)$ such that  $A\cap B\subseteq X_t$ and  $A\cap B\neq X_{t'}^{(1)}$ for all $t'\in W_2^{(1)}$ 
 is the same as in $\mathcal{S}$.
\end{itemize}
Then $\mathcal{S}'$ is an $H$-sequence.
\end{lemma}

\begin{proof}
Let $t\in W_2^{(1)}$ and let $(A,B)$ be a Whitney separation such that $A\cap B=X_t^{(1)}$. Consider a Whitney separation $(A',B')$ such that $A'\cap B'\neq X_t^{(1)}$. Then by the definition of the Tutte decomposition, either $A\subseteq A'$ or, symmetrically, $A\subseteq B'$. 

Let $t\in W_{\geq 3}^{(1)}$ and assume that $(A,B)$ is a Whitney separation such that $A\cap B\subseteq X_t^{(1)}$ but $A\cap B\neq X_{t'}^{(1)}$ for all $t'\in W_2^{(1)}$. 
Let also $(A',B')$ be a Whitney separation such that either $(A',B')=X_{t'}^{(1)}$ for some $t'\in W_2^{(1)}$ or $|(A'\cap B')\cap X_t^{(1)}|\leq 1$, i.e., $A'\cap B'$ doe not separate $X_t^{(1)}$.
Then again we have the same property that either $A\subseteq A'$ or, symmetrically, $A\subseteq B'$. 

These observations immediately imply that $\mathcal{S}'$ can be obtained from $\mathcal{S}$ by a series of swaps of consecutive Whitney switches $(A,B)$ and $(A',B')$ such that 
 either $A\subseteq A'$ or, symmetrically, $A\subseteq B'$. Then to show the lemma, we use the following straightforward observation.
Let $F$ be a 2-connected graph and let $(A,B)$ and $(A',B')$ be Whitney separations of $F$ such that $A\subseteq A'$. Then the graphs obtained by executing the two Whitney switches with respect to $(A,B)$ and $(A',B')$ in any order are identical. Recall also that switches with respect to $(A',B')$ and $(B',A')$  are equivalent. This implies that $\mathcal{S}'$ is an $H$-sequence.
\end{proof}

We show that we can restrict the set of considered Whitney switches. 

For $t\in W_{\geq 3}^{(1)}$, we say that $X_t^{(1)}$ is \emph{$\varphi$-good} if $G[X_t^{(1)}]$ is $\varphi$-isomorphic to $H[X_{\alpha(t)}^{(2)}]$, and $X_t^{(1)}$ is \emph{$\varphi$-bad} otherwise. Notice that if $G[X_t^{(1)}]$ is 3-connected, then $X_t^{(1)}$ is $\varphi$-good but this not always so if $G[X_t^{(1)}]$ is a cycle.  

\begin{figure}[ht]
\centering
\scalebox{0.6}{
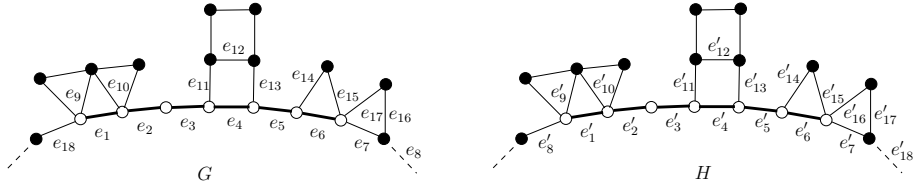}
\caption{An example of a $\varphi$-good segment; $\varphi(e_i)=e_i'$ for $i\in\{1,\ldots,18\}$, the vertices of the segment  are white.} 
\label{fig:good-segm}
\end{figure}

Let $t\in W_{\geq 3}^{(1)}$ such that  $X_t^{(1)}$ is  $\varphi$-bad. Clearly, $G[X_t^{(1)}]$ is a cycle. Let $\{t_1,\ldots,t_s\}=N_{T^{(1)}}^2(t)$ and denote
$G_t=G[X_t^{(1)}\cup \bigcup_{i=1}^sX_{t_i}^{(1)}]$ and $H_{\alpha(t)}=H[X_{\alpha(t)}^{(2)}\cup\bigcup_{i=1}^sX_{\alpha(t_i)}^{(2)}]$.
Let $P=v_0\cdots v_r$ be a path in $G[X_t^{(1)}]$ and $e_i=v_{i-1}v_i$  for $i\in\{1,\ldots,r\}$. We say that $P$ a \emph{$\varphi$-good segment} of $X_t^{(1)}$ if the following holds (see Fig.~\ref{fig:good-segm} for an example):
\begin{itemize}
\item[(i)] the length of $P$ is at least 5,
\item[(ii)]  there is a path $P'=u_0\cdots u_r$ in $H[X_{\alpha(t)}^{(2)}]$ such that $u_{i-1}u_i=\varphi(e_i)$ for all $i\in\{1,\ldots,r\}$,   
\item[(iii)] for every $i\in \{1,\ldots,r\}$ and 
for every $t'\in W_{\geq 3}^{(1)}$ such that $X_t^{(1)}\cap X_{t'}^{(1)}=\{v_{i-1},v_i\}$, $X_{t'}^{(1)}$ is $\varphi$-good,
\item[(iv)] for every $i\in\{1,\ldots,r-1\}$, $\varphi(E_{G_t}(v_i))=E_{H_{\alpha(t)}}(u_i)$. 
\end{itemize}

\begin{figure}[ht]
\centering
\scalebox{0.6}{
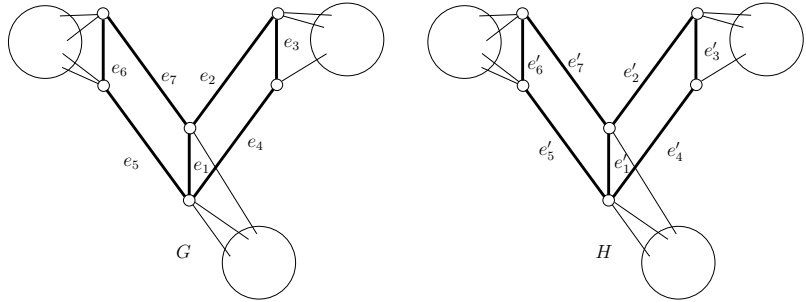}
\caption{Mutually $\varphi$-good bags; $\varphi(e_i)=e_i'$ for $i\in\{1,\ldots,7\}$, the vertices of the mutually $\varphi$-good bags of $G$ and the corresponding bags of $H$ are white.} 
\label{fig:mutually-good}
\end{figure}

For distinct $t_1,t_2\in W_{\geq 3}^{(1)}$ with a common neighbor in $T^{(1)}$, we say that $X_{t_1}^{(1)}$ and $X_{t_2}^{(1)}$ are \emph{mutually $\varphi$-good} (see Fig.~\ref{fig:mutually-good}) if they are $\varphi$-good
and 
$G[X_{t_1}^{(1)}\cup X_{t_2}^{(1)} ]$ is $\varphi$-isomorphic to $H[X_{\alpha(t_1)}^{(2)}\cup X_{\alpha(t_2)}^{(2)}]$.

We say that an $H$-sequence is \emph{$\varphi$-good} if no Whitney switch of $\mathcal{S}$ splits (mutually) $\varphi$-good bags and segments. Formally, for every switch with respect to some Whitney separation $(A,B)$  in $\mathcal{S}$, 
\begin{itemize}
\item[(i)] $X_t^{(1)}\subseteq A$ or $X_t^{(1)}\subseteq B$ for every $\varphi$-good bag $X_t^{(1)}$, 
\item[(ii)] $V(P)\subseteq A$ or $V(P)\subseteq B$ for every $\varphi$-good segment $P$,
\item[(iii)] $X_{t_1}^{(1)}\cup X_{t_2}^{(1)}\subseteq A$ or $X_{t_1}^{(1)}\cup X_{t_2}^{(1)}\subseteq B$
for every two distinct mutually $\varphi$-good bags  $X_{t_1}^{(1)}$ and $X_{t_2}^{(1)}$. 
\end{itemize}

We prove that it is sufficient to consider $\varphi$-good $H$-sequences.

\begin{lemma}\label{lem:good}
There is a minimum $H$-sequence of Whitney switches $\mathcal{S}$ that is $\varphi$-good.
\end{lemma}

\begin{proof}
First, we show that there is a minimum $H$-sequence of Whitney switches $\mathcal{S}$ such that for every switch with respect to some Whitney separation $(A,B)$  in $\mathcal{S}$, condition (i)
of the definition of $\varphi$-good sequences 
 is fulfilled.

Suppose that $\mathcal{S}$ is a minimum $H$-sequence of Whitney switches such that the number of switches with respect to Whitney separations $(A,B)$ that split $\varphi$-good bags is minimum. We show that $\mathcal{S}$ satisfies (i). The proof is by contradiction. Assume that there is $t\in W_{\geq 3}^{(1)}$ such that $X_t^{(1)}$ is $\varphi$-good 
and
there is a switch in $\mathcal{S}$ with respect to a Whitney separation $(A,B)$
with the property that $X_t^{(1)}\setminus A\neq\emptyset$ and $X_t^{(1)}\setminus B\neq\emptyset$.  By Lemma~\ref{lem:rearrange}, we can assume that all the switches with respect to Whitney separations $(A,B)$ such that $A\cap B\subseteq X_t^{(1)}$ 
splitting $X_t^{(1)}$
are in the end of $\mathcal{S}$ and denote by $\mathcal{S}'$ the subsequence of these switches.
Denote by $G'$ the graph obtained from $G$ by performing the switches prior $\mathcal{S}'$.  Then the graph $G''$ isomorphic to $H$ is obtained by performing $\mathcal{S}'$. 

Clearly, $G'[X_t^{(1)}]$ is a cycle of length at least 4. Denote by $v_1,\ldots,v_r$ the vertices of the cycle (in the cycle order) and let $e_i=v_{i-1}v_i$ for $i\in\{1,\ldots,r\}$ assuming that $v_0=v_r$ (i.e., the indices are taken modulo $r$). Notice that for each neighbor $t'$ of $t$ in $T^{(1)}$, $t'\in W_2^{(1)}$ and $X_{t'}^{(1)}=\{v_{i-1},v_i\}$ for some $i\in\{1,\ldots,r\}$. Assume that $N_T(t)=\{t_1,\ldots,t_s\}$ where $X_{t_i}^{(1)}=\{v_{j_i-1},v_{j_i}\}$ for  $1\leq j_1<\ldots<j_s\leq r$.  Denote by $T_1^{(1)},\ldots,T_s^{(1)}$ the subtrees of $T^{(1)}-t$ containing $t_1,\ldots,t_s$ respectively, and let $G_i$ be the subgraph of $G'$ induced by the vertices of $\bigcup_{h\in V(T^{(1)}_i)}X_h^{(1)}$ for $i\in\{1,\ldots,s\}$. For $i\in \{1,\ldots,s\}$, let $T_1^{(2)},\ldots,T_s^{(2)}$ be the subtrees of $T^{(2)}-\alpha(t)$ that contain $\alpha(t_1),\ldots,\alpha(t_s)$ respectively, and let $H_i$ be the subgraph of $H$ induced by the vertices of $\bigcup_{h\in V(T^{(2)}_i)}X_h^{(2)}$ for $i\in\{1,\ldots,s\}$. 

Since $X_t^{(1)}$ is $\varphi$-good, $\varphi(e_1),\ldots,\varphi(e_r)$ form a cycle of $H$ in the given order. Assume that $\varphi(e_i)=u_{i-1}u_i$ for $i\in\{1,\ldots,r\}$ for $u_1,\ldots,u_r$ forming $X_{\alpha(t)}^{(2)}$ (assuming that $u_0=u_r$).
Observe that the graph $\varphi$-isomorphic to $H$ is obtained from $G'$ by Whitney switches with respect to Whitney separations $(A,B)$ such that $A\cap B=\{v_i,v_j\}$ for distinct nonadjacent 
vertices $v_i$ and $v_j$ for some $i,j\in\{1,\ldots,r\}$. This implies that $G_i$ is $\varphi$-isomorphic to $H_i$ for every $i\in\{1,\ldots,s\}$ as the switches do not affect these graphs. 
However, $G'$ and $H$ are not $\varphi$-isomorphic by the minimality of $\mathcal{S}$. 
By Lemma~\ref{lem:folk}, we obtain that there are $i\in\{1,\ldots,r\}$ such that $\varphi(E_{G'}(v_i))\neq 
E_{H}(u_i)$. More precisely, taking into account that every $G_i$ is $\varphi$-isomorphic to $H_i$, we have that there is $i\in\{1,\ldots,s\}$ such that 
$\varphi(E_{G_i}(v_{j_i}))=E_{H_i}(u_{j_i-1})$ and $\varphi(E_{G_i}(v_{j_i-1}))=E_{H_i}(u_{j_i})$. 
Denote by $I\subseteq \{1,\ldots,s\}$ the set of all such indices $i\in\{1,\ldots,s\}$.  

We define the partially signed circular permutation $\pis^c=(\langle 1,s_1\rangle,\ldots,\langle r,s_r\rangle)$ such that 
$s_{j_i}=-1$ for all $i\in I$, $s_{j_i}=+1$ for all $i\in \{1,\ldots,s\}\setminus I$, and $s_j=0$ for $j\in \{1,\ldots,r\}\setminus \{j_1,\ldots,j_s\}$. The crucial observation is that obtaining the graph $\varphi$-isomorphic to $H$ from $G'$ by Whitney switches is equivalent to sorting $\pis^c$ by reversals. By Lemma~\ref{lem:signs}, there is an optimal sorting sequence composed by trivial reversals $\rhos^c(j,j)$ for $s_j=-1$. This corresponds to performing the Whitney switches with respect to Whitney separations $(V(G_i),(V(G')\setminus V(G_i))\cup\{v_{j_i-1}v_{j_i}\} )$ for all $i\in I$. This contradicts the choice of $\mathcal{S}$, because these switches do not split $X_t^{(1)}$.

By the next step, we show  that there is a minimum $H$-sequence of Whitney switches $\mathcal{S}$ such that for every switch with respect to some Whitney separation $(A,B)$  in $\mathcal{S}$, conditions (i) and (ii) of the definition of $\varphi$-good sequences hold. The proof is similar to the first part.

Suppose that $\mathcal{S}$ is a minimum $H$-sequence of Whitney switches that satisfies (i) such that the number of switches with respect to Whitney separations $(A,B)$ that split $\varphi$-good segments is minimum. We claim that $\mathcal{S}$ satisfies (ii). Assume that this is not the case and there is a switch that splits some $\varphi$-good segments $P$. Assume that $P$ is a path of 
$G[X_t^{(1)}]$ for $t\in W_{\geq 3}^{(1)}$. Denote by $P_1,\ldots,P_\ell$ the family of inclusion maximal $\varphi$-good segments in $G[X_t^{(1)}]$; note that $P$ is a subpath of one of these segments.
 By Lemma~\ref{lem:rearrange}, we can assume that all the switches with respect to Whitney separations $(A,B)$ such that (a) $A\cap B\subseteq X_t^{(1)}$, (b) either $(A,B)$ splits $X_t^{(1)}$ or $A\cap B=\{x,y\}$ and $xy$ is an edge of one of the paths $P_1,\ldots,P_\ell$ are in the end of $\mathcal{S}$, and denote by $\mathcal{S}'$ the subsequence of these switches.
Denote by $G'$ the graph obtained from $G$ by performing the switches prior $\mathcal{S}'$.  Then the graph $G''$ isomorphic to $H$ is obtained by performing $\mathcal{S}'$.

Denote by $v_1,\ldots,v_r$ the vertices of the cycle $G'[X_t^{(1)}]$ (in the cycle order) and let $e_i=v_{i-1}v_i$ for $i\in\{1,\ldots,r\}$ assuming that $v_0=v_r$ (i.e., the indices are taken modulo $r$). Notice that for each neighbor $t'$ of $t$ in $T^{(1)}$, $t'\in W_2^{(1)}$ and $X_{t'}{(1)}=\{v_{i-1},v_i\}$ for some $i\in\{1,\ldots,r\}$. Assume that $N_T(t)=\{t_1,\ldots,t_s\}$ where $X_{t_i}^{(1)}=\{v_{j_i-1},v_{j_i}\}$ for  $1\leq j_1<\ldots<j_s\leq r$.  Denote by $T_1^{(1)},\ldots,T_s^{(1)}$ the subtrees of $T^{(1)}-t$ containing $t_1,\ldots,t_s$ respectively, and let $G_i$ be the subgraph of $G'$ induced by the vertices of $\bigcup_{h\in V(T^{(1)}_i)}X_h^{(1)}$ for $i\in\{1,\ldots,s\}$. For $i\in \{1,\ldots,s\}$, let $T_1^{(2)},\ldots,T_s^{(2)}$ be the subtrees of $T^{(2)}-\alpha(t)$ that contain $\alpha(t_1),\ldots,\alpha(t_s)$ respectively, and let $H_i$ be the subgraph of $H$ induced by the vertices of $\bigcup_{h\in V(T^{(2)}_i)}X_h^{(2)}$ for $i\in\{1,\ldots,s\}$.

Notice that $G_i$ is $\varphi$-isomorphic to $H_i$ for every $i\in\{1,\ldots,s\}$. If $e_{j_i}$ is an edge of one of the paths $P_1,\ldots,P_\ell$, then this follows from the conditions (iii) and (iv) of the definition of $\varphi$-good segments. Otherwise, $\mathcal{S}'$ does not contain a switch with respect to a separation $(A,B)$ with $A\cap B=\{v_{j_i-1},v_{j_i}\}$, that, $G_i$ is not affected by switches in $\mathcal{S}'$. 

 Denote by $e_1',\ldots,e_r'$ the edges of $H[X_{\alpha(t)}^{(2)}]$ taken in the cycle order and denote by $u_1,\ldots,u_r$ the vertices of this cycle such that $e_i'=u_{i-1}u_i$ (assuming that $u_0=u_r$). For $i\in\{1,\ldots,s\}$, let $\varphi(e_{j_i})=e_{j_i'}$ for $j_1',\ldots,j_s'\in\{1,\ldots,r\}$. Since each $G_i$ is $\varphi$-isomorphic to $H_i$, we have that, by Lemma~\ref{lem:folk}, for every $i\in\{1,\ldots,s\}$, either $\varphi(E_{G_i}(v_{j_i-1}))=E_{G_i}(u_{j_i'-1})$ and $\varphi(E_{G_i}(v_{j_i}))=E_{G_i}(u_{j_i'})$ or, symmetrically,  $\varphi(E_{G_i}(v_{j_i-1}))=E_{G_i}(u_{j_i'})$ and $\varphi(E_{G_i}(v_{j_i}))=E_{G_i}(u_{j_i'-1})$. Let $I=\{i\mid 1\leq i\leq s, \varphi(E_{G_i}(v_{j_i-1}))=E_{G_i}(u_{j_i'-1})\text{ and }\varphi(E_{G_i}(v_{j_i}))=E_{G_i}(u_{j_i'})\}$, and let $\bar{I}=\{1,\ldots,r\}\setminus I$.
 
 We construct the following partially signed circular permutation $\pis^c=(\langle \pi_1,s_1\rangle,\ldots,\langle \pi_r,s_r\rangle)$ such that for every $i\in \{1,\ldots,r\}$, $e_{\pi_i}=\varphi^{-1}(e_i')$. For $i\in \{1,\ldots,s\}$, we define $s_{j_i}=+1$ if $i\in I$ and $s_{j_i}=-1$ if $i\in\bar{I}$. The other sign are zeros, that is, $s_j=0$ if $j\notin \{i_1,\ldots,i_s\}$. 
Notice that by the definition of $\varphi$-good segments, $(\langle \pi_i,s_i\rangle,\ldots,\langle \pi_j,s_j\rangle )$ (with indices taken modulo $r$) is a signed block of $\pis^c$ of length at least 5 if and only if the edges $e_{\pi_1},\ldots,e_{\pi_j}$ form a $\varphi$-good segment.  

Similarly to the first part of the proof, we have  that obtaining the graph $\varphi$-isomorphic to $H$ from $G'$ by Whitney switches is equivalent to sorting $\pis^c$ by reversals. 
By Lemma~\ref{lem:strip-cut-circ}, there is an optimal sorting sequence that does not cut strips of $\pis^c$ of length at least 5. This implies, that there is a sequence of Whitney switches $\mathcal{S}''$ of the same length as $\mathcal{S}'$ such that applying $\mathcal{S}''$ to $G'$ creates a graph $G''$ isomorphic to $H$ and for every switch $\mathcal{S}''$ with with respect a separation $(A,B)$, either $V(P_i)\subseteq A$ or $V(P_i)\subseteq B$  for every $i\in\{1,\ldots,\ell\}$. This contradicts the choice of $\mathcal{S}$ and, therefore, proves the claim.

Finally, we show that every minimum $H$-sequence of Whitney switches $\mathcal{S}$ satisfying conditions (i) and (ii) of the definition of $\varphi$-good sequences satisfies (iii) as well.

Let $t\in W_2^{(1)}$.  By Lemma~\ref{lem:rearrange}, we can assume that all the switches with respect to Whitney separations $(A,B)$ such that $A\cap B=X_t^{(1)}$ are executed in the end of sequence. Denote by $\mathcal{S}'$ the subsequence of $\mathcal{S}$ formed by these Whitney switches. In the same way as before, let $G'$ be the graph obtained from $G$ by performing the switches prior $\mathcal{S}'$.  Then the graph $G''$ isomorphic to $H$ is obtained by performing $\mathcal{S}'$. Notice that for every neighbor $t'$ of $t$ in $T^{(1)}$, the bag $X_{t'}^{(1)}$ is $\varphi$-good in $G'$. Moreover, because $\mathcal{S}$ satisfies (i), for every distinct $t_1,t_2\in N_{T^{(1)}}(t)$ such that $X_{t_1}^{(1)}$ and $X_{t_2}^{(1)}$ are mutually $\varphi$-good for the initial graph $G$, these bags are mutually $\varphi$-good for $G'$. This implies that if $\mathcal{S}'$ is empty, then the claim hold. Assume that $\mathcal{S}'$ contains at least one  switch. 

Since every  $X_{t'}^{(1)}$ is $\varphi$-good in $G'$ for $t'\in N_{T^{(1)}}(t)$, there is the partition $(N_1,N_2)$ of $N_{T^{(1)}}(t)$ such that $X_{t_1}^{(1)}$ and $X_{t_2}^{(1)}$ are mutually $\varphi$-good in $G'$ for distinct $t_1,t_2\in N_{T^{(1)}}(t)$ if and only if either $t_1,t_2\in N_1$ or $t_1,t_2\in N_2$. This implies that $\mathcal{S}'$ consist of the single Whitney switch for the unique Whitney separation $(A,B)$ such that 
$\bigcup_{t'\in N_1}X_{t'}^{(1)}\subseteq A$, $\bigcup_{t'\in N_2}X_{t'}^{(1)}\subseteq B$ and $A\cap B=X_{t}^{(1)}$. Clearly, this separation does not split any 
$X_{t_1}^{(1)}$ and $X_{t_2}^{(1)}$ that are mutually $\varphi$-good for the initial graph $G$. This means that (iii) is fulfilled.
\end{proof}

Let $t\in W_{\geq 3}^{(1)}$ be such that $X_t^{(1)}$ is  $\varphi$-bad.
 Denote by  $t_1,\ldots,t_s\neq t$ the nodes of $N_{T^{(1)}}^2(t)$. Let 
$G_t=G[X_t^{(1)}\cup\bigcup_{i=1}^s X_{t_i}^{(1)}]$ and $H_{\alpha(t)}=G[X_{\alpha(t)}^{(2)}\cup\bigcup_{i=1}^s X_{\alpha(t_i)}^{(2)}]$. In words, $G_t$ is the subgraphs of $G$ induced by the vertices of $X_t^{(1)}$ and the vertices of the bags at distance two in $T^{(1)}$ from $t$, and $H_{\alpha(t)}$ the subgraph of $H$ induced by the vertices of the bags that are images of the bags composing $G_t$ according to $\alpha$. 

We say that a vertex $v\in X_t^{(1)}$ is a \emph{crucial breakpoint} if $\varphi(E_{G_t}(v))\neq E_{H_{\alpha(t)}}(u)$ for every $u\in V(H_{\alpha(t)})$. We denote by $b(G)$ the total number of crucial breakpoints in the $\varphi$-bad bags and say that $b(G)$ is the \emph{breakpoint number} of $G$.
Recall that by our convention, $G$ and $H$ are enhanced graphs, but we extend this definition for the general case needed in the next section. For (not necessarily enhanced) 2-isomorphic graphs $G$ and $H$, and a 2-isomorphism $\varphi$, we construct 
their enhancements $\widehat{G}$ and $\widehat{H}$, and consider the enhanced mapping $\widehat{\varphi}$. Then $b(G)$ is defined as $b(\widehat{G})$. 

Observe that if $G$ and $H$ are $\varphi$-isomorphic, then $b(t)=0$ by Lemma~\ref{lem:folk}, but not the other way around. 

We conclude the section by giving a lower bound for the length of an $H$-sequence.

\begin{lemma}\label{lem:lower-break}
Let $\mathcal{S}$ be an $H$-sequence of Whitney switches. Then $b(G)/2\leq |\mathcal{S}|$.
\end{lemma}

\begin{proof}
The claim is trivial if $b(G)=0$. Assume that $b(G)>0$.
Let $\mathcal{S}$ be an $H$-sequence of Whitney switches, that is, the graph $G'$ obtained from $G$ by applying $\mathcal{S}$ is $\varphi$-isomorphic to $H$. By Lemma~\ref{lem:folk}, $b(G')=0$. Hence, $\mathcal{S}$ should contain switches that decrease the breakpoint number. Clearly,  the Whitney switch with respect to a Whitney partition $(A,B)$ reduces $b(G)$ if and only if $A\cap B=\{u,v\}$, where at least one of $u$ or $v$ is a crucial breakpoint. 
 Then the switch decreases $b(G)$ by at most 2 and the claim follows.
\end{proof}

\section{Kernelization for \probWS}\label{sec:kern}
In this section, we show that \probWS parameterized by $k$ admits a polynomial kernel. To do it, we obtain a more general result by proving that the problem has a polynomial kernel when parameterized by the breakpoint  number of the first input graph.

\begin{theorem}\label{thm:kern}
\probWS has a kernel such that each graph in the obtained instance has at most $\max\{52\cdot b-36,3\}$ vertices, where $b$ is the breakpoint number of the input graph. 
\end{theorem}

\begin{proof}
Let $(G,H,\varphi,k)$ be an instance of \probWS, where $G$ and $H$ are $n$-vertex 2-connected 2-isomorphic graphs, $\varphi\colon E(G)\rightarrow E(H)$ is a 2-isomorphism, and $k$ is a nonnegative integer.  

First, we use Proposition~\ref{prop:tutte} to construct the Tutte decompositions of $G$ and $H$. Denote by 
$\mathcal{T}^{(1)}=(T^{(1)},\{X_t^{(1)}\}_{t\in V(T^{(1)})})$
and  $\mathcal{T}^{(2)}=(T^{(2)},\{X_t^{(2)}\}_{t\in V(T^{(2)})})$  the constructed Tutte decompositions of $G$ and $H$ respectively, and let $(W_2^{(h)},W_{\geq 3}^{(h)})$ be the partition of $V(T^{(h)})$ satisfying (T4)--(T8) for $h=1,2$.

In the next step, we construct the isomorphism $\alpha\colon V(T^{(1)})\rightarrow V(T^{(2)})$ satisfying conditions (i)--(iii) of Lemma~\ref{lem:isom-decomp}. Recall that Lemma~\ref{lem:isom-decomp} claims that such an isomorphism always exists.
If $T^{(1)}$ and $T^{(2)}$ are single-vertex trees, then the construction is trivial. Assume that this is not the case. 
Let $t$ be a leaf of $T^{(1)}$ and let $t'$ be its unique neighbor. We have that $t\in W_{\geq 3}^{(1)}$ and $E(G[X_t^{(1)}])\setminus E(G[X_{t'}^{(1)}])\neq \emptyset$. By (iii), we have that 
there is a unique leaf $t''$ of $T^{(2)}$ such that $\varphi(E(G[X_t^{(1)}])\setminus E(G[X_{t'}^{(1)}]))\subseteq E(H[X_{t''}^{(2)}])$ and $\alpha(t)=t''$. 
Clearly, $t''$ can be found in polynomial time. This means that we can construct in polynomial time the restriction of $\alpha$ on the leaves of $T^{(1)}$ that maps them bijectively on the leaves of $T^{(2)}$. Since $T^{(1)}$ and $T^{(2)}$ are isomorphic, there is a unique way to extend $\alpha$ from leaves to $V(T^{(1)})$. This can be done by picking a \emph{root} node $r$ of $T^{(1)}$ and computing $\alpha$ bottom-up starting from the leaves. Given that $\alpha$ is already computed for the leaves, the construction of $\alpha$ can be completed in $\Oh(|V(T^{(1)})|)$ time. 

Given $\alpha$, we compute the enhancements $\widehat{G}$ and $\widehat{H}$ of $G$ and $H$ respectively, and then define the enhanced mapping $\widehat{\varphi}\colon E(\widehat{G})\rightarrow E(\widehat{H})$. Clearly, this can be done in polynomial time. Note that $\alpha$ satisfies the conditions of Lemma~\ref{lem:isom-decomp-enh}.
Observe also that we can verify in polynomial time whether a bag $X_t^{(1)}$ for $t\in W_{\geq 3}^{(1)}$ is $\varphi$-good or not. Then  we can compute in polynomial  time  $b(G)=b(\widehat{G})$. 

To simplify notation, let $G:=\widehat{G}$, $H:=\widehat{H}$ and $\varphi:=\hat{\varphi}$.

Now we apply a series of reduction rules that are applied for $G$, $H$, $\varphi$, and the Tutte decompositions of $G$ and $H$.

The aim of the first rule is to decrease the total size of bags that are  $\varphi$-bad (see Fig.~\ref{fig:rule-one} for an example).

\begin{figure}[ht]
\centering
\scalebox{0.6}{
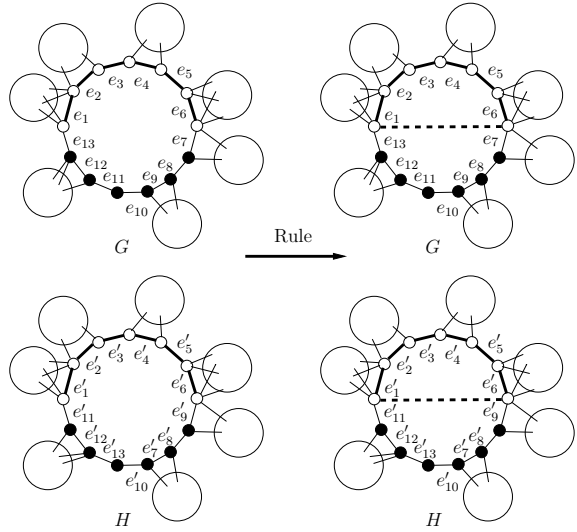}
\caption{An example of an application of Reduction Rule~\ref{red:segments}; $\varphi(e_i)=e_i'$ for $i\in\{1,\ldots,13\}$, the vertices of the $\varphi$-good segment in $G$ and the corresponding segment in $H$ are white, and the added edges are shown by dashed lines.} 
\label{fig:rule-one}
\end{figure}

\begin{redrule}\label{red:segments}
If for $t\in W_{\geq 3}^{(1)}$ such that $X_t^{(1)}$ is  $\varphi$-bad, there is an inclusion maximal $\varphi$-good segment $P=v_0\cdots v_r$, then do the following:
\begin{itemize}
\item find the path $P'=u_0\cdots u_r$ in $H[X_{\alpha(t)}^{(2)}]$ composed by the edges $u_{i-1}u_i=\varphi(v_{i-1}v_i)$ for $i\in\{1,\ldots,r\}$,
\item add the edge $v_0v_r$ to $G$ and $u_0u_r$ to $H$, 
\item extend $\varphi$ by setting $\varphi(v_0u_r)=u_0u_r$,
\item recompute the Tutte decompositions of the obtained graphs and the isomorphism $\alpha$.
\end{itemize}
\end{redrule}

\begin{claim}\label{cl:r-segment}
Reduction Rule~\ref{red:segments} is safe, does not increase the breakpoint number, and can be executed in polynomial time. 
\end{claim}

\begin{proof}[Proof of Claim~\ref{cl:r-segment}]
Denote by $\tilde{G}$ the graph obtained from $G$ by the application of Reduction Rule~\ref{red:segments} for $P=v_0\cdots v_r$. Let also $\tilde{H}$ be the graph obtained from $H$ and denote by $\tilde{\varphi}$ in the extension of $\varphi$. Since $\varphi$ maps the edges of $P$ into the edges of $P'$, we have that $\tilde{\varphi}$ is a 2-isomorphism of $\tilde{G}$ to $\tilde{H}$. 

Suppose that $(G,H,\varphi,k)$ is a yes-instance of \probWS. By Lemma~\ref{lem:good}, there is a $\varphi$-good  $H$-sequence of Whitney switches of length at most $k$ that transforms $G$ into the graph $G'$ that is $\varphi$-isomorphic to $H$. By condition (ii) of the definition of a $\varphi$-good $H$-sequence, for every Whitney switch in $\mathcal{S}$, it is preformed with respect to a Whitney separation $(A,B)$ such that either $v_0,\ldots,v_r\in A$ or $v_0,\ldots,v_r\in B$. This implies that $\mathcal{S}$ can be performed on $\tilde{G}$ and 
transforms $\tilde{G}$ into $\tilde{G}'$ that is $\tilde{S}$-isomorphic to $\tilde{H}$. This means that  $(\tilde{G},\tilde{H},\tilde{\varphi},k)$ is a yes-instance. 

It is straightforward to see that every sequence of Whitney switches transforming $\tilde{G}$ into a graph $\tilde{\varphi}$-isomorphic to $\tilde{H}$ can be applied to $G$ and produces the graph $\varphi$-isomorphic to $H$. Therefore, if  $(\tilde{G},\tilde{H},\tilde{\varphi},k)$ is a yes-instance of \probWS, then $(G,H,\varphi,k)$ is a yes-instance as well.

To show that $b(\tilde{G})=b(G)$, we explain how to recompute the Tutte decompositions. For this, observe that we add a chord to a cycle of $G$ that forms a bag of the Tutte decomposition. This operation splits the bag into two bags of size at least 3 and the bag of size 2 composed by the end-vertices of the chord. Formally, this is done as follows.
We replace $t$ in $T^{(1)}$ by three nodes $t_1$, $t_2$ and $t'$, and define the corresponding bags $X_{t_1}^{(1)}=\{v_0,\ldots,v_r\}$, $X_{t_2}^{(1)}=X_t^{(1)}\setminus \{v_1,\ldots,v_{r-1}\}$, and $X_{t'}^{(1)}=\{v_0,v_r\}$. Notice that for every $t''\in N_{V(T^{(1)})}(t)$, either $X_{t''}^{(1)}\subseteq X_{t_1}^{(1)}$ or  $X_{t''}^{(1)}\subseteq X_{t_2}^{(1)}$. In the first case, we make $t''$ adjacent to $t_1$ and $t''$ is adjacent to $t_2$ in the second case. We modify $T^{(2)}$ and redefine $\alpha$ in similar way. The node $\alpha(t)$ is replaced by three nodes $\alpha(t_1)$, $\alpha(t_2)$ and $\alpha(t')$, with $X_{\alpha(t_1)}^{(2)}=\{u_0,\ldots,u_r\}$, $X_{\alpha(t_2)}^{(2)}=X_{\alpha(t)}^{(2)}\setminus \{u_1,\ldots,u_{r-1}\}$, and $X_{\alpha(t')}^{(2)}=\{u_0,u_r\}$. For every $t''\in N_{V(T^{(2)})}(\alpha(t))$, either $X_{t''}^{(2)}\subseteq X_{\alpha(t_1)}^{(2)}$ or  $X_{t''}^{(2)}\subseteq X_{\alpha(t_2)}^{(2)}$. We make $t''$ adjacent to $\alpha(t_1)$ in the first case and $t''$ is adjacent to $\alpha(t_2)$ in the second case. It is straightforward to verify that we obtain the Tutte decompositions of $\tilde{G}$ and $\tilde{H}$ respectively, and the obtained $\alpha$ is an isomorphism of the modified tree $T^{(1)}$ to the modified tree $T^{(2)}$ satisfying the conditions of Lemma~\ref{lem:isom-decomp-enh}. Notice that the vertices of $X_{t'}^{(1)}$ are adjacent and the same holds for $X_{\alpha(t')}^{(2)}$, that is, $\tilde{G}$ and $\tilde{H}$ are enhanced.

Thus, we obtain two bags $X_{t_1}^{(1)}$ and $X_{t_2}^{(1)}$ of size at least 3 from $X_t^{(1)}$ and both of them induce cycles. Since $P$ is a $\varphi$-good segment and $\tilde{\varphi}(v_0v_r)=u_0u_r$, we have that $X_{t_1}^{(1)}$ is a $\tilde{\varphi}$-good. Moreover, for every $i\in\{0,\ldots,r\}$,
$\tilde{\varphi}(E_{\tilde{G}[X_{t_1}^{(1)}]}(v_i))= E_{\tilde{H}[X_{\alpha(t_1)}^{(2)}]}(u_i)$. This implies that the number of crucial breakpoints does not increase.

To argue that Reduction Rule~\ref{red:segments} can be applied in polynomial time, observe first that inclusion maximal $\varphi$-good segments can be recognized in polynomial time. For each $t\in W_{\geq 3}^{(1)}$, we can verify whether $X_t^{(1)}$ is $\varphi$-good in polynomial time using Lemma~\ref{lem:folk}. Then for each $t\in W_{\geq 2}^{(1)}$ such that $X_t^{(1)}$ is a $\varphi$-bad bag, we consider all at most $n^2$ paths $P$ of the cycle $G[X_t^{(1)}]$ and for each $P$, we verify conditions (i)--(iv) of the definition of a $\varphi$-segment. It is easy to see that each of these conditions can be verified in polynomial time. Further, given an inclusion maximal $\varphi$-good segment $P$, we can apply the rule in polynomial time. 
Note also that we can avoid recomputing the Tutte decompositions of $\tilde{G}$ of $\tilde{H}$ from scratch as the described above computation procedure can be done in polynomial time.
\end{proof}

Reduction Rule~\ref{red:segments} is applied exhaustively while we are able to find $\varphi$-good segments. To simplify notation, we use $G$, $H$ and $\varphi$ to denote the obtained graphs and the obtained 2-isomorphism. We also keep the notation used for the Tutte decompositions.

Our next reduction rule is used to simplify the structure of $\varphi$-good bags by turning them into cliques (see Fig.~\ref{fig:rule-two} for an example). 

\begin{figure}[ht]
\centering
\scalebox{0.6}{
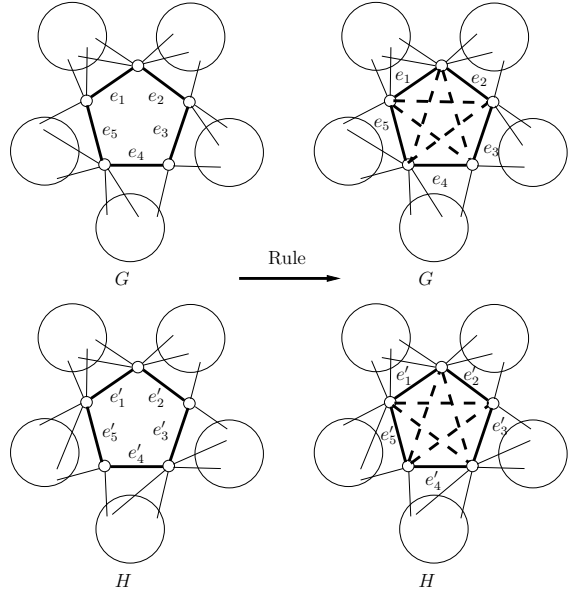}
\caption{An example of an application of Reduction Rule~\ref{red:good-bags}; $\varphi(e_i)=e_i'$ for $i\in\{1,\ldots,5\}$, the vertices of the $\varphi$-good bag in $G$ and the corresponding bag of $H$ are white, and the added edges are shown by dashed lines.} 
\label{fig:rule-two}
\end{figure}

\begin{redrule}\label{red:good-bags}
If for $t\in W_{\geq 3}^{(1)}$ such that $X_t^{(1)}$ is a $\varphi$-good, there are nonadjacent vertices in $X_t^{(1)}$, then compute the $\varphi$-isomorphism $\psi$ of $G[X_t^{(1)}]$ to $H[X_{\alpha(t)}^{(2)}]$ and for every nonadjacent 
$u,v\in X_t^{(1)}$, do the following:
\begin{itemize}
\item add the edge $uv$ to $G$ and $\psi(u)\psi(v)$ to $H$, 
\item extend $\varphi$ by setting $\varphi(uv)=\psi(u)\psi(v)$.
\end{itemize}
\end{redrule}

\begin{claim}\label{cl:r-good-bags}
Reduction Rule~\ref{red:good-bags} is safe, does not change the Tutte decompositions and the breakpoint number, and can be executed in polynomial time. 
\end{claim}

\begin{proof}[Proof of Claim~\ref{cl:r-good-bags}]
Let $t\in W_{\geq 3}^{(1)}$ be such that $X_t^{(1)}$ is $\varphi$-good and there are nonadjacent vertices in $X_t^{(1)}$. 
Recall that $G[X_{t}^{(1)}]$ and $H[X_{\alpha(t)}^{(2)}]$ are $\varphi$-isomorphic by the definition of $\varphi$-good bags. Therefore, there is $\varphi$-isomorphism $\psi$ of $G[X_t^{(1)}]$ to $H[X_{\alpha(t)}^{(2)}]$. 

Denote by $\tilde{G}$ the graph obtained from $G$ by the application of one step of Reduction Rule~\ref{red:good-bags} for two nonadjacent $u,v\in X_t^{(1)}$, that is, $\tilde(G)$ is obtained by adding $uv$ to $G$. Let $\tilde{H}$ be the graph obtained from $H$ by adding $\psi(u)\psi(v)$, and let $\tilde{\varphi}$ be the extension of $\varphi$ on $uv$. Since $\psi$ is a $\varphi$-isomorphism, we conclude that $\tilde{\varphi}$ is a 2-isomorphism of $\tilde{G}$ to $\tilde{H}$. 

 Suppose that $(G,H,\varphi,k)$ is a yes-instance of \probWS. By Lemma~\ref{lem:good}, there is a $\varphi$-good  $H$-sequence of Whitney switches that transforms $G$ into the graph $G'$ that is $\varphi$-isomorphic to $H$. By condition (i) of the definition of a $\varphi$-good $H$-sequence, for every Whitney switch in $\mathcal{S}$, it is preformed with respect to a Whitney separation $(A,B)$ such that either $X_t^{(1)}\subseteq A$ or $X_t^{(1)}\subseteq B$. Therefore, $\mathcal{S}$ can be performed on $\tilde{G}$ and 
transforms $\tilde{G}$ into $\tilde{G}'$ that is $\tilde{S}$-isomorphic to $\tilde{H}$. This means that  $(\tilde{G},\tilde{H},\tilde{\varphi},k)$ is a yes-instance. 
The opposite claim, that if  $(\tilde{G},\tilde{H},\tilde{\varphi},k)$ is a yes-instance of \probWS, then $(G,H,\varphi,k)$ is a yes-instance as well, is straightforward, because 
every sequence of Whitney switches transforming $\tilde{G}$ into a graph $\tilde{\varphi}$-isomorphic to $\tilde{H}$ can be applied to $G$ and produces the graph $\varphi$-isomorphic to $H$. 

This proves that the rule is safe. To show the remaining claims,  observe that the rule transforms $X_t^{(1)}$ and $X_{\alpha(t)}^{(2)}$ into cliques and does not affect other bags. 
Moreover, for every $v\in X_t^{(1)}$, $\tilde{\varphi}(E_{\tilde{G}[X_t^{(1)}]})=E_{\tilde{H}[X_{\alpha(t)}^{(2)}]}(\psi(v))$.
Therefore, the rule
does not change the Tutte decompositions and the breakpoint number.  
For every $t\in W_{\geq 3}^{(1)}$, we can verify whether $X_t^{(1)}$ is $\varphi$-good in polynomial time using Lemma~\ref{lem:folk}. By the same lemma, we can compute $\psi$ in polynomial time. Therefore, Reduction Rule~\ref{red:good-bags} can be applied in polynomial time.
\end{proof}

We apply Reduction Rule~\ref{red:good-bags} for all bags of $G$ that are not cliques. We use the same convention as for the first rule, and keep the old notation for the obtained graphs, their Tutte decompositions, and the obtained 2-isomorphism.
 
The next aim is to reduce the number of mutually $\varphi$-good bags by ``gluing'' them into cliques (see Fig.~\ref{fig:rule-three} for an example). 

\begin{figure}[ht]
\centering
\scalebox{0.6}{
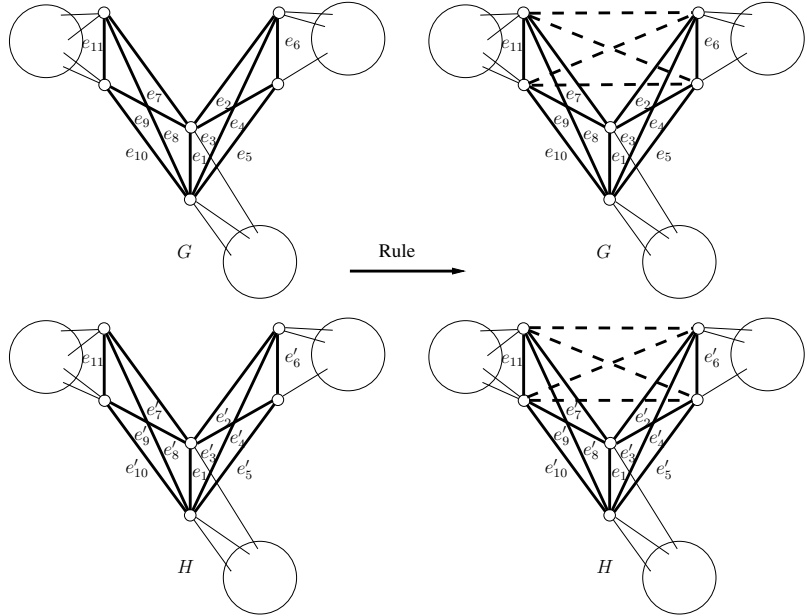}
\caption{An example of an application of Reduction Rule~\ref{red:mutually-good-bags}; $\varphi(e_i)=e_i'$ for $i\in\{1,\ldots,11\}$, the vertices of the mutually $\varphi$-good bags of $G$ and the corresponding bags of $H$ are white, and the added edges are shown by dashed lines.} 
\label{fig:rule-three}
\end{figure}

\begin{redrule}\label{red:mutually-good-bags}
For distinct $t_1,t_2\in W_{\geq 3}^{(1)}$ such that $X_{t_1}^{(1)}$ and $X_{t_2}^{(1)}$ are mutually $\varphi$-good, 
\begin{itemize}
\item compute the $\varphi$-isomorphism $\psi$ of $G[X_{t_1}^{(1)}\cup X_{t_2}^{(1)} ]$ to $H[X_{\alpha(t_1)}^{(2)}\cup X_{\alpha(t_2)}^{(2)} ]$,
\item  for every $u\in X_{t_1}^{(1)}\setminus X_{t_2}^{(1)}$ and every $v\in X_{t_2}^{(1)}\setminus X_{t_1}^{(1)}$, do the following:
\begin{itemize}
\item add the edge $uv$ to $G$ and $\psi(u)\psi(v)$ to $H$, 
\item extend $\varphi$ by setting $\varphi(uv)=\psi(u)\psi(v)$,
\end{itemize}
\item recompute the Tutte decompositions of the obtained graphs and the isomorphism $\alpha$.
\end{itemize}
\end{redrule}

\begin{claim}\label{cl:r-mutually-good-bags}
Reduction Rule~\ref{red:mutually-good-bags} is safe, does not change the breakpoint number, and can be executed in polynomial time. 
\end{claim}

\begin{proof}[Proof of Claim~\ref{cl:r-mutually-good-bags}]
The proof of safeness essentially repeats the proof for Reduction Rule~\ref{red:good-bags}. 

Let  $t_1,t_2\in W_{\geq 3}^{(1)}$ such that $X_{t_1}^{(1)}$ and $X_{t_2}^{(1)}$ are mutually $\varphi$-good. By the definition of mutually $\varphi$-good bags, 
$G[X_{t_1}^{(1)}\cup X_{t_2}^{(1)} ]$ is $\varphi$-isomorphic to $H[X_{\alpha(t_1)}^{(2)}\cup X_{\alpha(t_2)}^{(2)} ]$ and, therefore, $\psi$ exists. 

Denote by $\tilde{G}$ the graph obtained from $G$ by the addition of one edge $uv$, and denote by $\tilde{H}$ the graph obtained from $H$ by adding $\psi(u)\psi(v)$. Let also $\tilde{\varphi}$ be the extension of $\varphi$ on $uv$.  Because $\psi$ is a $\varphi$ isomorphism of $G[X_{t_1}^{(1)}\cup X_{t_2}^{(1)} ]$ to $H[X_{\alpha(t_1)}^{(2)}\cup X_{\alpha(t_2)}^{(2)} ]$, $\tilde{\varphi}$ is a 2-isomorphism of $\tilde{G}$ to $\tilde{H}$. 

Suppose that $(G,H,\varphi,k)$ is a yes-instance of \probWS. By Lemma~\ref{lem:good}, there is a $\varphi$-good  $H$-sequence of Whitney switches of length at most $k$ that transforms $G$ into the graph $G'$ that is $\varphi$-isomorphic to $H$. By condition (iii) of the definition of a $\varphi$-good $H$-sequence, for every Whitney switch in $\mathcal{S}$, it is preformed with respect to a Whitney separation $(A,B)$ such that either $X_{t_1}^{(1)}\cup X_{t_2}^{(1)} \subseteq A$ or $X_{t_1}^{(1)}\cup X_{t_2}^{(1)}\subseteq B$. Therefore, $\mathcal{S}$ can be performed on $\tilde{G}$ and 
transforms $\tilde{G}$ into $\tilde{G}'$ that is $\tilde{S}$-isomorphic to $\tilde{H}$. Hence,  $(\tilde{G},\tilde{H},\tilde{\varphi},k)$ is a yes-instance. 
The opposite implication is straightforward. We conclude that the rule is safe.

To recompute the Tutte decompositions, observe that by Reduction Rule~\ref{red:good-bags}, $X_{t_1}^{(1)}$, $X_{t_2}^{(1)}$,  $X_{\alpha(t_1)}^{(2)}$ and $X_{\alpha(t_2)}^{(2)}$  are cliques and, therefore, 
Reduction Rule~\ref{red:mutually-good-bags} makes cliques from  $X_{t_1}^{(1)}\cup X_{t_2}^{(1)}$ and $X_{\alpha(t_1)}^{(2)}\cup X_{\alpha(t_2)}^{(2)}$. Hence, to recompute the Tutte decompositions of $G$ and $H$, we have to identify the nodes $t_1$ and $t_2$ of $T^{(1)}$ and the nodes $\alpha(t_1)$ and $\alpha(t_2)$ of $T^{(2)}$, respectively. Every neighbor of $t_1$ or $t_2$ in $T^{(1)}$ (every neighbor of  $\alpha(t_1)$ or $\alpha(t_2)$, respectively) distinct from these nodes becomes the neighbor of the obtained node $t$ ($\alpha(t)$, respectively). Recall that there is $t'\in W_2^{(1)}$ such that $X_{t_1}^{(1)}\cap X_{t_2}^{(1)}=X_{t'}^{(1)}$. If $X_{t'}^{(1)}$ is not a separator of the graph constructed by the rule (i.e., if $t'$ has exactly two neighbors $t_1$ and $t_2$ in the original tree $T^{(1)}$), then we delete $t'$ and $\alpha(t')$ from the trees obtained from $T^{(1)}$ and $T^{(2)}$, respectively.  
It is straightforward to verify that this procedure, indeed, recomputes the Tutte decompositions and $\alpha$.

Since Reduction Rule~\ref{red:good-bags} does not affect the bags that are  $\varphi$-bad and 
 for every $v\in X_{t'}^{(1)}$, $\tilde{\varphi}(E_{\tilde{G}[X_{t}^{(1)}]})=E_{\tilde{H}[X_{\alpha(t)}^{(2)}]}(\psi(v))$, we have that the breakpoint number remains the same. 

To show that the rule can be executed in polynomial time, note that  we can verify whether $X_{t_1}^{(1)}$ and $X_{t_2}^{(1)}$ are mutually $\varphi$-good
and then compute
$\psi$ in polynomial time using Lemma~\ref{lem:folk}. Clearly, recomputing the Tutte decomposition can be done in polynomial time. Then the total running time is polynomial.
\end{proof}

Reduction Rule~\ref{red:mutually-good-bags} is applied exhaustively whenever it is possible. As before, we do not change the notation   for the obtained graphs, their Tutte decompositions, and the obtained 2-isomorphism.

Our next rule is used to perform the Whitney switches that are unavoidable. To state the rule, we define the following auxiliary instance of \probWS. Let $C^{(1)}$ and $C^{(2)}$ be copies of $C_4$ with the edges $e_1,e_2,e_3,e_4$ and $e_1',e_2',e_3',e_4'$, respectively, taken in the cycle order. We define $\chi(e_1)=e_1$, $\chi(e_2)=e_4'$, $\chi(e_3)=e_3'$ and $\chi(e_4)=e_2'$. Clearly, $\chi$ is a 2-isomorphism of $C^{(1)}$ to $C^{(2)}$, but $C^{(1)}$ and $C^{(2)}$ are not $\chi$-isomorphic. This means that $I=(C^{(1)},C^{(2)},\chi,0)$ is a no-instance of \probWS. We call this instance the \emph{trivial} no-instance. Notice that for each no-instance, the input graphs should have at least 4 vertices each. Therefore, $I$ is a no-instance of minimum size. 

\begin{figure}[ht]
\centering
\scalebox{0.6}{
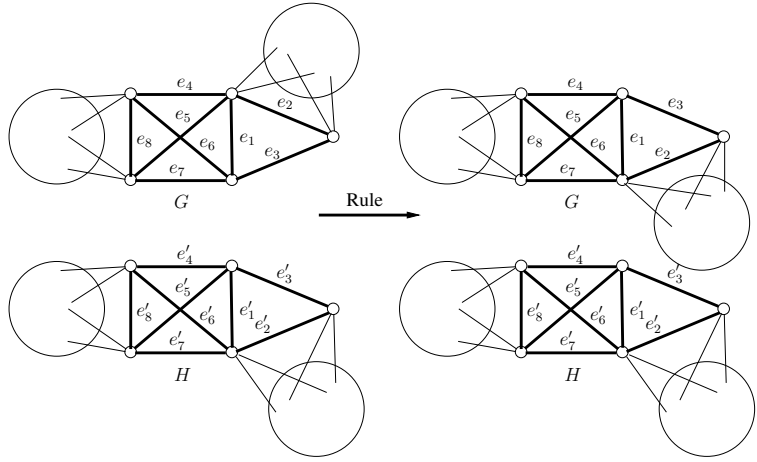}
\caption{An example of an application of Reduction Rule~\ref{red:switch-good}; $\varphi(e_i)=e_i'$ for $i\in\{1,\ldots,8\}$, the vertices of the switched $\varphi$-good bags in $G$ and the corresponding bags of $H$ are white.} 
\label{fig:rule-four}
\end{figure}

\begin{redrule}\label{red:switch-good}
If there is $t\in W_2^{(1)}$ such that $d_{T^{(1)}}(t)=2$ and for the neighbors $t_1$ and $t_2$ of $t$, it holds that $X_{t_1}^{(1)}$ and $X_{t_2}^{(2)}$ are $\varphi$-good but not mutually $\varphi$-good, then do the following:
\begin{itemize}
\item find the connected components $T_1$ and $T_2$ of $T^{(1)}-t$, and construct $A=\bigcup_{t'\in V(T_1)}X_{t'}^{(1)}$ and $B=\bigcup_{t'\in V(T_2)}X_{t'}^{(1)}$,
\item perform the Whitney switch with respect to the separation $(A,B)$,
\item set $k:=k-1$, and if $k<0$, then return the trivial no-instance and stop.
\end{itemize}
\end{redrule}

An example is shown in Fig.~\ref{fig:rule-four}.

\begin{claim}\label{cl:r-switch-good}
Reduction Rule~\ref{red:switch-good} is safe, does not change the Tutte decompositions and the breakpoint number, and can be executed in polynomial time. 
\end{claim}

\begin{proof}[Proof of Claim~\ref{cl:r-switch-good}]
Suppose that $(G,H,\varphi,k)$ is a yes-instance of \probWS. By Lemma~\ref{lem:good}, there is a $\varphi$-good  $H$-sequence of Whitney switches of length at most $k$ that transforms $G$ into the graph $G'$ that is $\varphi$-isomorphic to $H$. 

We claim that  $\mathcal{S}$ contains the Whitney switch with respect to $(A,B)$. Suppose that this is not the case. For every Whitney separation $(A',B')$ in $\mathcal{S}$, by condition (ii) of the definition of $\varphi$-food $H$-sequences, 
 we have that  either $X_{t_1}^{(1)}\subseteq A'$ or $X_{t_1}^{(1)}\subseteq B'$ and either $X_{t_2}^{(1)}\subseteq A'$ or $X_{t_2}^{(1)}\subseteq B'$. There is no separation $(A'',B'')$ in $\mathcal{S}$ with $A''\cap B''=X_{t_1}^{(1)}\cap X_{t_1}^{(2)}$, because $(A,B)$ is the unique separation with this property. It follows that for 
every Whitney separation $(A',B')$ in $\mathcal{S}$, either $X_{t_1}^{(1)}\cup X_{t_2}^{(1)}\subseteq A'$ or $X_{t_1}^{(1)}\cup X_{t_2}^{(2)}\subseteq B'$. However, because 
$X_{t_1}^{(1)}$ and $X_{t_2}^{(2)}$ are not mutually $\varphi$-good, $G[X_{t_1}^{(1)}\cup X_{t_2}^{(1)}]$ is not $\varphi$-isomorphic to $H[X_{\alpha(t_1)}^{2)}\cup X_{\alpha(t_2)}^{(2)}]$. This  contradicts that $\mathcal{H}$ is an $H$-sequence. 

Since $\mathcal{S}$ contains the Whitney switch with respect to $(A,B)$, then by Lemma~\ref{lem:rearrange}, we can assume that this switch is first in $\mathcal{S}$. Let $\tilde{G}$ be the graph obtained from $G$ by performing this switch. Then $(\tilde{G},H,\varphi,k-1)$ is a yes-instance of \probWS. Clearly, $k\geq 1$ in this case and we not stop by Reduction Rule~\ref{red:switch-good}.

Let $\tilde{G}$ be the graph obtained from $G$ by performing the Whitney switch with respect to $(A,B)$. Trivially, if $(\tilde{G},H,\varphi,k-1)$ is a yes-instance of \probWS, then 
$(G,H,\varphi,k)$ is a yes-instance. This completes the safeness proof.

Clearly, the Whitney switch with respect to $(A,B)$ does not change the Tutte decomposition. Also, the switch does not affect the bags that are  $\varphi$-bad and, moreover, 
if $X_t^{(1)}$ is  $\varphi$-bad and $t'$ is at distance two in $T^{(1)}$ from $t$, then $G[X_t^{(1)}\cup X_{t'}^{(1)}]$ is not modified.
 Therefore, the breakpoint number remains the same.
 
We can verify for every $t\in W_2^{(1)}$ such that $d_{T^{(1)}}(t)=2$, whether  for the neighbors $t_1$ and $t_2$ of $t$, it holds that $X_{t_1}^{(1)}$ and $X_{t_2}^{(2)}$ are $\varphi$-good but not mutually $\varphi$-good in polynomial time by Lemma~\ref{lem:folk}.
\end{proof}

Reduction Rule~\ref{red:switch-good} is applied exhaustively whenever it is possible. Note that after applying this rule, we are able to apply   
Reduction Rule~\ref{red:mutually-good-bags} and we do it. 

Suppose that the algorithm did not stop while executing Reduction Rule~\ref{red:switch-good}.
In the same way as with previous rules, we maintain the initial notation  for the obtained graphs, their Tutte decompositions, and the obtained 2-isomorphism.

Our final rule deletes simplicial vertices of degree at least 3.

\begin{redrule}\label{red:delete}
If there is a simplicial vertex $v\in V(G)$ with $d_G(v)\geq 3$, then do the following:
\begin{itemize}
\item find the vertex $u\in V(H)$ such that $E_H(u)=\varphi(E_G(v))$,
\item set $G:=G-v$ and $H:=H-u$, 
\item set $\varphi:=\varphi|_{E(G)\setminus E_G(v)}$.
\end{itemize}
\end{redrule}

\begin{claim}\label{cl:r-delete}
Reduction Rule~\ref{red:delete} is safe  and can be executed in polynomial time. Moreover, the rule does not increase the breakpoint number and 
 the Tutte decompositions of the obtained by the rule graphs are constructed by the deletions of $v$ and $u$ from the bags of the Tutte decompositions of $G$ and $H$, respectively.
\end{claim}

\begin{proof}[Proof of Claim~\ref{cl:r-delete}]
Let $v$ be a simplicial vertex of $G$ of degree at least 3. 

Because $N_G[v]$ is a clique of size at least 4, three is a unique $t\in W_{\geq 3}^{(1)}$ such that $v$ is a simplicial vertex of $G[X_t^{(1)}]$ and $v\notin X_{t'}^{(1)}$ for every $t'\in V(T^{(1)})$ distinct from $t$. Recall that after the exhaustive application of Reduction Rules~\ref{red:good-bags}--\ref{red:switch-good}, the bags of the Tutte decompositions of  $G$ and $H$, respectively, are cliques. In particular, this means that 
$G[X_t^{(1)}]$ and $G[X_{\alpha(t)}^{(2)}]$ are 3-connected and, therefore, $\varphi$-isomorphic by Lemma~\ref{lem:isom-decomp-enh}. Then there is $u\in X_{\alpha(t)}^{(1)}$ such that 
$\varphi(E_{G[X_t^{(1)}]}(v))=E_{G[X_{\alpha(t)}^{(2)}]}(u)$. Moreover, $u$ does not belong to any other bag of the Tutte decomposition of $H$ except $X_{\alpha(t)}^{(2)}$. This implies that $u$ is a slimplicial vertex of $H$ and $E_H(u)=\varphi(E_G(v))$. This means that, given $v$, there is unique $u\in V(H)$ such that $E_H(u)=\varphi(E_G(v))$.

Let $\tilde{G}=G-v$ and $\tilde{H}=H-u$.
Since $X_{t}^{(1)}$ and $X_{\alpha(t)}^{(2)}$ are cliques and $v$ and $u$ do not belong to any separator of size 2 of $G$ and $H$, respectively, $\tilde{G}$ and $\tilde{H}$ are 2-connected. Since 
 $E_H(u)=\varphi(E_G(v))$, we have that  $\tilde{\varphi}=\varphi|_{E(G)\setminus E_G(v)}$ is a 2-isomorphism of $\tilde{G}$ to $\tilde{H}$. Observe also that the Tutte decompositions of $\tilde{G}$
and $\tilde{H}$ are obtained by the deletion of $v$ and $u$ from $X_t^{(1)}$ and $X_{\alpha(t)}^{(2)}$, respectively, and this proves the last part of the claim. 
Since vertex deletion can only decrease the breakpoint number, $b(\tilde{G})\leq b(G)$.

Now we show that $(G,H,\varphi,k)$ is yes-instance of \probWS if and only if $(\tilde{G},\tilde{H},\tilde{\varphi},k)$ is a yes-instance.

Since $X_t^{(1)}$ is a clique, for every Whitney separation $(A,B)$ of $G$, $v\notin A\cap B$ and
ether $X_t^{(1)}\subseteq A$ or $X_t^{(1)}\subseteq B$. Let $\mathcal{S}$ be an $H$-sequence.
 We modify this sequence as follows. For every Whitney separation $(A,B)$ used in $\mathcal{S}$, we replace it by the separation $(A\setminus \{v\},B\setminus\{v\})$.
Denote by $\tilde{\mathcal{S}}$ the obtained sequence. Then $\tilde{\mathcal{S}}$ is an $\tilde{H}$-sequence.
This means that if $(G,H,\varphi,k)$ is yes-instance of \probWS, then $(\tilde{G},\tilde{H},\tilde{\varphi},k)$ is a yes-instance.

For the opposite direction, notice that for every Whitney separation $(A,B)$ of $\tilde{G}$, 
ether $X_t^{(1)}\setminus \{v\}\subseteq A$ or $X_t^{(1)}\setminus\{v\}\subseteq B$.
Let $\tilde{\mathcal{S}}$ be an $\tilde{H}$-sequence. For every Whitney separation $(A,B)$ used in $\tilde{\mathcal{S}}$, we replace it by the separation $(A\cup \{v\},B)$ if $X_t^{(1)}\setminus \{v\}\subseteq A$ and by $(A,B \cup \{v\})$ otherwise. Then we have that the obtained sequence $\mathcal{S}$ is an $H$-sequence. Therefore, 
if $(\tilde{G},\tilde{H},\tilde{\varphi},k)$ is a yes-instance, then $(G,H,\varphi,k)$ is yes-instance. 

To complete the proof, it remains to observe that a simplicial vertex $v$ can be recognized in polynomial time, and we can find the corresponding vertex $u$ by checking whether 
$E_H(u)=\varphi(E_G(v))$ in polynomial time. Then the rule can be applied in polynomial time.
\end{proof}

Reduction Rule~\ref{red:delete} is applied exhaustively. Let $G$, $H$ and $\varphi$ be the resulting graphs. We also keep the same notation for the Tutte decompositions of $G$ and $H$ and the isomorphism $\alpha$ following the previous convention. This completes the description of our kernelization algorithm as the graphs $G$ and $H$ have bounded size. More precisely, we show the following claim.

\begin{claim}\label{cl:bound} 
$|V(G)|=|V(H)|\leq \max\{52\cdot b(G)-36,3\}$.
\end{claim}

\begin{proof}[Proof of Claim~\ref{cl:bound}]
Clearly, $|V(G)|=|V(H)|$ and it is sufficient to show that $|V(G)|\leq \max\{52\cdot b(G)-36,3\}$.

First, we show that if $W_2^{(1)}\neq\emptyset$, then for every $t\in W_2^{(1)}$, there is a neighbor $t'$ in $T^{(1)}$ such that $X_{t'}^{(1)}$ is $\varphi$-bad.
Suppose that this is not the case and there is  $t\in W_2^{(1)}$ with the neighbors $t_1,\ldots,t_s$ such that $X_{t_i}^{(1)}$ is $\varphi$-good for every $i\in \{1,\ldots,s\}$.
Note that $s\geq 2$ by the definition of the Tutte decomposition. Since Reduction Rule~\ref{red:mutually-good-bags} is not applicable, 
for every distinct $i,j\in\{1,\ldots,s\}$, $X_{t_i}^{(1)}$ and $X_{t_j}^{(1)}$ are not mutually $\varphi$-good. This implies that $s=2$, but then we are able to apply Reduction Rule~\ref{red:switch-good}, a contradiction. 

Suppose that $b(G)=0$. We show that $|V(G)|=3$.
For this, we observe that $|V(T^{(1)})|=1$, because otherwise,  there is $W_2^{(1)}\neq\emptyset$, and  for $t\in W_2^{(1)}$, we have that $X_{t'}^{1}$ is $\varphi$-good for every neighbor $t'$ of $t$ in $T^{(1)}$. If $|X_t^{(1)}|\geq 4$, we would be able to apply Reduction Rule~\ref{red:delete}. We conclude that $|X_t^{(1)}|=3$ and $|V(G)|=3$.

Assume from now that $b(G)\geq 1$, that is, the Tutte decomposition has bags that are $\varphi$-bad. Let $W'\subseteq W_{\geq 3}^{(1)}$ and $W''\subseteq W_{\geq 3}^{(1)}$ be the sets of $t\in W_{\geq 3}^{(1)}$ such that $X_t^{(1)}$ are $\varphi$-good and $\varphi$-bad, respectively. Note that $W''\neq\emptyset$ but $W'$ may be empty.
Denote $U=\bigcup_{t\in W''}X_t^{(1)}$, that is, $U$ is the set of vertces of the bags that are $\varphi$-bad.

We claim that for every $t\in W'$, ether $t$ is a leaf of $T^{(1)}$ or $X_t^{(1)}\subseteq U$. 
Suppose that $t\in W'$ is not a leaf of $T^{(1)}$. Let $t_1,\ldots,t_s\in W_2^{(1)}$ be the neighbors of $t$ in $T^{(1)}$. Since $t$ is not a leaf, $s\geq 2$. 
Let $Z=\bigcup_{i=1}^sX_{t_i}^{(1)}$. 
Assume that $X_t^{(1)}\setminus Z\neq \emptyset$. Since $s\geq 2$, $|Z|\geq 3$. Therefore, $X_t^{(1)}$ is a clique of size at least 4 and every $v\in X_t^{(1)}\setminus Z$ is a simplicial vertex of $G$. However, in this case, we would be able to apply Reduction Rule~\ref{red:delete}, a contradiction. Therefore $X_t^{(1)}=Z$.
We proved that for every $i\in\{1,\ldots,s\}$, $t_i$ has a neighbor $t_i'$ in $T^{(1)}$ such that $t_i'\in W''$. Since $X_{t_i}^{(1)}\subseteq X_{t_i'}^{(1)}$ for $i\in\{1,\ldots,s\}$,
 $X_t^{(1)}\subseteq Z\subseteq \bigcup_{i=1}^sX_{t_i'}^{(1)}\subseteq U$ as required. 

Let $t\in W'$ be a leaf of $T^{(1)}$. We prove that $|X_t^{(1)}\setminus U|=1$. Since $t$ is a leaf, there is the unique neighbor $t'$ of $t$ in $T^{(1)}$. We proved that $t'\in W_{2}^{(1)}$ has a neighbor $t''$ such that $X_{t''}^{(1)}$ is $\varphi$-bad. Thus, $X_{t'}^{(1)}\subseteq U$. Suppose that $|X_t^{(1)}\setminus X_{t}^{(1)}|\geq 2$. In this case, $X_{t}^{(1)}$ is a clique of size at least 4 and we would be able to apply  Reduction Rule~\ref{red:delete} for $v\in X_t^{(1)}\setminus X_{t}^{(1)}$, because this is a simplicial vertex of $G$.  This cannot happen and we have that $|X_t^{(1)}\setminus X_{t'}^{(1)}|=1$. Therefore, $|X_t^{(1)}\setminus U|=1$.

Suppose that $t\in W_{2}^{(1)}$. We show that $t$ has at most two neighbors $t'$ in $T^{(1)}$ such that $t'$ is a leaf and $X_{t'}^{(1)}$ is $\varphi$-good.
To obtain a contradiction, assume that $t_1,\ldots,t_s$ are the neighbors of $t$  that are leaves of $T^{(1)}$, $X_{t_i}^{(1)}$ is $\varphi$-good  for each $i\in\{1,\ldots,s\}$, and $s\geq 3$. But then there are distinct $i,j\in\{1,\ldots,s\}$ such that $X_{t_i}^{(1)}$ and $X_{t_j}^{(1)}$ are mutually $\varphi$-good. Then we would be able to apply Reduction Rule~\ref{red:mutually-good-bags}, a contradiction.

Summarizing all these observations, we obtain that the number of leaves $t$ of $T^{(1)}$ such that $X_t^{(1)}$ is $\varphi$-good is at most $2|E[U]|$ and $|V(G)\setminus U|\leq 2|E(G[U])|$. 
Since $U$ is composed by cycles formed by $\varphi$-bad bags, $|E(G[U])|\geq |V(G[U])|$. Hence, to upper bound the number of vertices of $G$, it is sufficient to upper bound the number of edges of $G[U]$.
 
Let $T'$ be the forest obtained from $T^{(1)}$ by the deletion of the nodes $t\in W'$ and then the nodes $t'\in W_2^{(1)}$ that became leaves. Consider $t\in V(T')$ that is a leaf of $T'$ or an isolated node. If $t$ is an isolated node, then $X_t^{(1)}$ contains a crucial breakpoint that is not contained in other bags $X_{t'}^{(1)}$ for $t'\in W''$. If $t$ is a leaf, then there is the unique $t'\in W_2^{(1)}$ that is the neighbor of $t$ in $T'$. Since $X_t^{(1)}$ is $\varphi$-bad, $X_t^{(1)}$ has a crucial breakpoint $v$ such that $v\notin X_{t'}^{(1)}$. This means, that $v$ is not in any $X_{t''}^{(1)}$ for $t''\in W''$. Applying these arguments inductively, we obtain that $|W''|\leq b(G)$. Then the 
number of edged $R$ of $G[U]$ that a included in at least two $\varphi$-bad bags is at most $b(G)-1$. 
Moreover, we can observe the following. For each $t\in W''$, $G[X_t^{(1)}]-R$ is either a cycle (if no edge of $R$ is an edge of $G[X_t^{(1)}]$) or a union of vertex disjoint paths. 
Denote by $\mathcal{P}$ the family of such cycles and paths taken over all $t\in W''$. Then we have that $|\mathcal{P}|\leq 2(b(G)-1)$. Let $\mathcal{P}\rq{}$ be the family off all inclusion maximal subpaths of the elements of $\mathcal{P}$ that does not have crucial breakpoint as internal vertices. We obtain that
$|\mathcal{P}\rq{}|\leq |\mathcal{P}|+b(G)\leq 3b(G)-2$. 

We claim that the total length of the paths of $\mathcal{P}\rq{}$ is at most $12b(G)-8$. To obtain a contradiction, assume that the total length is at least $12b(G)-7$. Then by the pigeonhole principle, there is a path $P\in \mathcal{P}\rq{}$ of length at least 5. Let $P=v_0\cdots v_r$ and assume that $P$ is a segment of the cycle $X_t^{(1)}$ for $t\in W\rq{}\rq{}$. 
 Let $\{t_1,\ldots,t_s\}=N_{T^{(1)}}^2(t)$ and denote
$G_t=G[X_t^{(1)}\cup \bigcup_{i=1}^sX_{t_i}^{(1)}]$ and $H_{\alpha(t)}=H[X_{\alpha(t)}^{(2)}\cup\bigcup_{i=1}^sX_{\alpha(t_i)}^{(2)}]$.
Since $P$ does not contain edges of $R$, for every $i\in \{1,\ldots,r\}$ and
for every $t\rq{}\in W_{\geq 3}^{(1)}$ such that $X_t^{(1)}\cap X_{t\rq{}}^{(1)}$, $X_{t\rq{}}^{(1)}$ is $\varphi$-good. Since $v_1,\ldots,v_{r-1}$ are not crucial breakpoints, there is a path $P\rq{}=u_0\cdots u_r$ in $H[X_{\alpha{1}}^{(1)}]$ such that $u_{i-1}u_i=\varphi(v_{i-1}v_i)$ for every $i\in\{1,\ldots,r\}$
and $\varphi(E_{G_t}(v_i))=E_{H_{\alpha(t)}}(u_i)$ for $i\in\{1,\ldots,r-1\}$.
It follows that $P$ is a $\varphi$-good segment of $X_t^{(1)}$. However, this means that we should be able to apply Reduction Rule~\ref{red:segments}, a contradiction. 

Since the total length of paths of $\mathcal{P}\rq{}$ is at most $12\cdot b(G)-8$, we obtain that $G[U]$ has at most $13\cdot b(G)-9$ edges by taking into account the edges of $R$. 
Then $|V(G)|\leq 4|E(G[U])|=52\cdot b(G)-36$ and this completes the proof.
\end{proof}

Recall that Reduction Rules~\ref{red:segments}--\ref{red:delete} do not increase the breakpoint number. Therefore, for the the obtained instance $(G,H,\varphi,k)$ of \probWS, $|V(G)|=|V(H)|\leq \min\{52\cdot b-36,3\}$, where $b$ is the breakpoint number of the initial input graph $G$.

Finally, we have to argue that the kernelization algorithm is polynomial. For this, recall that the intimal construction of the Tutte decompositions of the input graphs and the isomorphism $\alpha$ is done in polynomial time. Further, we apply Reduction Rules~\ref{red:segments}--\ref{red:delete}, and 
we proved that each of then can be done in polynomial time in Claims~\ref{cl:r-segment}--\ref{red:delete}, respectively. By each application of one of Reduction Rules~\ref{red:segments}--\ref{red:switch-good}, we add at least one edge. Therefore, the rules are executed at most $n^2$ times. By Reduction Rule~\ref{red:delete}, we delete one vertex. Then the rule is called at most $n$ times. This implies that the total running time is polynomial.
\end{proof}

Theorem~\ref{thm:kern} implies that \probWS has a polynomial kernel when parameterized by $k$ and we can show Theorem~\ref{thm:main} that we restate.  

\medskip
\noindent
{\bf Theorem~\ref{thm:main}.}
{\it    
\probWS  admits a kernel with $\Oh(k)$ vertices and is solvable in  $2^{\Oh(k\log k)}\cdot n^{\Oh(1)}$ time.
}

\begin{proof}
Let $(G,H,\varphi,k)$ be an instance of \probWS. We compute the breakpoint number $b(G)$. If $b(G)>2k$, then by Lemma~\ref{lem:lower-break}, $(G,H,\varphi,k)$ is a no-instance. In this case, we return the trial no-instance of \probWS (defined in the proof of Theorem~\ref{thm:kern}) and stop.  Otherwise, we use the kernelization algorithm from Theorem~\ref{thm:kern} that returns an instance, where   
each of the graphs has at most $\max\{104\cdot k-36,3\}$ vertices.

Combining the kernelization with the brute-force checking of at most $k$ Whitney switches immediately leads to the algorithm running in $2^{\Oh(k\log k)}\cdot n^{\Oh(1)}$ time.
\end{proof}

In Corollary~\ref{cor:hard}, we proved that \probWS is \classNP-hard when the input graphs are constrained to be cycles. Theorem~\ref{thm:kern} indicates that it is the presence of bags in the Tutte decompositions that are cycles of length at least 4 that makes \probWS difficult, because only such cycles may contain crucial breakpoint. In particular, we can derive the  following straightforward corollary.

\begin{corollary}\label{cor:no-break}
Let $(G,H,\varphi,k)$ be an instance of \probWS such that $b(G)=0$. Then \probWS for this  instance can be solved in polynomial time.
\end{corollary}

For example, the condition that $b(G)=0$ holds when $G$ and $H$ have no induced cycles of length at least 4, that is, when $G$ and $H$ are chordal graphs.

\begin{corollary}\label{cor:chordal} 
\probWS can be solved in polynomial time on chordal graphs.
\end{corollary}

\section{Conclusion}\label{sec:concl}
We proved that \probWS admits a polynomial kernel when parameterized by the breakpoint number of the input graphs and this implies that the problem has a polynomial kernel when parameterized by $k$. More precisely, we obtain a kernel, where the graphs have $\Oh(k)$ vertices. Using this kernel, we can solve \probWS in  $2^{\Oh(k\log k)}\cdot n^{\Oh(1)}$ time. It is natural to ask whether the problem can be solved in a single-exponential in $k$ time.

Another interesting direction of research is to investigate approximability for \probWS. In~\cite{BermanK99}, Berman and Karpinski proved that for every $\varepsilon>0$, it is \classNP-hard to approximate the reversal distance $d(\pi)$ for a linear permutation $\pi$  within factor $\frac{1237}{1236}-\varepsilon$. This result can be translated for circular permutations and this allows to obtain inapproximability lower bound for \probWS on cycles similarly to Corollary~\ref{cor:hard}. 
From the positive side, the currently best $1.375$-approximation for $d(\pi)$ was given by Berman, Hannenhalli, and Karpinski~\cite{BermanHK02}. Due to the close relations between \probWS and the sorting by reversal problem, it is interesting to check whether the same approximation ratio can be achieved for \probWS.

In \probWS, we are given two graphs $G$ and $H$ together with a 2-isomorphism and the task is to decide whether we can apply at most $k$ Whitney switches to obtain a graph $G'$ from $G$ such that $G'$ is $\varphi$-isomorphic to $H$. We can relax the task and ask whether we can obtain $G'$ that is isomorphic to $H$, that is, we do not require an isomorphism of $G$ to $H$ be a $\varphi$-isomorphism. Formally, we define the following problem.

\defproblema{\probUWS}{$2$-Isomorphic graphs $G$ and $H$, and a nonnegative integer $k$. }{Decide whether it is possible to obtain a graph $G'$  from $G$ by at most $k$ Whitney switches such that $G'$ is isomorphic to $H$.}

Note that if $\varphi$ is a 2-isomorphism of $G$ to $H$, then the minimum number of Whitney switches needed to obtain $G'$ that is $\varphi$-isomorphic to $H$ gives an upper bound for the number of Whitney switches required to obtain from $G$ a graph that isomorphic to $G$. However, these values can be arbitrary far apart. Consider two cycles $G$ and $H$ with the same number of vertices.  Clearly, $G$ and $H$ are isomorphic but for a given 2-isomorphism $\varphi$ of $G$ to $H$, we may need many Whitney switches to obtain $G'$ that is $\varphi$-isomorphic to $H$ and the number of switches is not bounded by any constant.

Using Proposition~\ref{prop:perm-hard}, we can show that \probUWS is \classNP-hard for very restricted instances.

\begin{proposition}\label{prop:hard-unlabeled}
\probUWS is \classNP-complete when restricted to 2-connected series-parallel graphs even if the input graphs are given together with their $2$-isomorphism.
\end{proposition} 

\begin{proof}
By Proposition~\ref{prop:perm-hard}, it is \classNP-complete to decide for given a circular permutation $\pi^c$ and a nonnegative integer $k$, whether $d^c(\pi^c)\leq k$. We reduce from this problem.

\begin{figure}[ht]
\centering
\scalebox{0.7}{
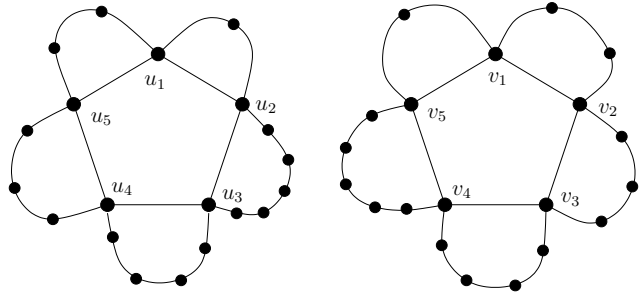}
\caption{Construction of $G$ and $H$ for $\pi^c=(2,1,5,4,5)$.}
\label{fig:hardness}
\end{figure}

Let $\pi^c=(\pi_1,\ldots,\pi_n)$ be a circular permutation. We construct the graph $G$ as follows:
\begin{itemize}
\item construct an $n$-vertex cycle $C=u_0u_1\cdots u_n$ assuming that $u_0=u_n$,
\item for every $i\in\{1,\ldots,n\}$, construct a $(u_{i-1},u_i)$-path $P_i$ of length $\pi_i+1$.
\end{itemize}
 The graph $H$ is constructed in the same way for $\iota^c=(1,\ldots,n)$, that is, we do the following:
  \begin{itemize}
\item construct an $n$-vertex cycle $C'=v_0v_1\cdots v_n$ assuming that $v_0=v_n$,
\item for every $i\in\{1,\ldots,n\}$, construct a $(v_{i-1},v_i)$-path $P_i'$ of length $i+1$.
\end{itemize}  
The construction of $G$ and $H$ is shown in Figure~\ref{fig:hardness}. 

We define $\varphi\colon E(G)\rightarrow E(H)$ as follows:
\begin{itemize}
\item for $i\in\{1,\ldots,n\}$, set $\varphi(u_{i-1}u_i)=u_{\pi_i-1}u_{\pi_i}$,
\item for $i\in\{1,\ldots,n\}$, $\varphi$ maps the edges of $P_i$ to the edges of $P_i'$ following the path order of the paths staring with the edges incident to $v_{i-1}$ and $u_{\pi_i-1}$, respectively.
\end{itemize}
It is straightforward to verify that $\varphi$ is a 2-isomorphism of $G$ to $H$.

We claim that $d^c(\pi^c)\leq k$ if and only if $(G,H,k)$ is a yes-instance of \probUWS.

Suppose that $d^c(\pi^c)\leq k$.  Then there is a sorting sequence $\mathcal{S}$ of circular reversals of length at most $k$ for $\pi^c$. We use the equivalence between reversals for circular permutations and Whitney switches on cycles described in Section~\ref{sec:reversals} and apply the equivalent to $\mathcal{S}$ sequence $\mathcal{S}'$ of Whitney switches for $C$. It is easy to see that $\mathcal{S}'$ produces the graph isomorphic to $H$.

For the opposite direction, assume that $(G,H,k)$ is a yes-instance of \probUWS. Then there is a sequence of Whitney switches $\mathcal{S}$ of length at most $k$ such that the graph $G'$ obtained from $G$ by applying $\mathcal{S}$ is isomorphic to $H$. Note that the Whitney switch with respect to any Whitney separation $(A,B)$ such that $A\cap B\subseteq V(P_i)$ for some $i\in\{1,\ldots,n\}$ results in a graph isomorphic to $G$. Therefore, we can assume that every Whitney switch in $\mathcal{S}$ is performed with respect to a Whitney separation $(A,B)$ such that $A\cap B$ is a pair of nonadjacent vertices of $C$. We again use the equivalence between circular reversal and Whitney switches on cycles and consider the sequence $\mathcal{S}'$ of circular reversals for $\pi^c$ that is equivalent to $\mathcal{S}$. Since each path $P_i$ has length $\pi_i+1$ and every switch from $\mathcal{S}$ does not affect $P_i$, we obtain that $\mathcal{S}'$ produce the identity circular  permutation $\iota^c$. Hence, $d^c(\pi^c)\leq k$.  
\end{proof}

Proposition~\ref{prop:hard-unlabeled} lead to the question about the parameterized complexity of \probUWS. In particular, does the problem admit a polynomial kernel when parameterized by $k$? 

Notice that to deal with  \probUWS, we should be able to check whether the input graphs $G$ and $H$ are isomorphic. If we are given a 2-isomorphism $\varphi$ of $G$ to $H$, then checking whether $G$ and $H$ are $\varphi$-isomorphic can be done in polynomial time by  Lemma~\ref{lem:folk}. However, checking whether $G$ and $H$ are isomorphic, even if a 2-isomorphism $\varphi$ is given, is a complicated task. For example, it can be observed that this is at least as difficult as solving \textsc{Graph Isomorphism} on tournaments (recall that a \emph{tournament} is a directed graph such that for every two distinct vertices $u$ and $v$, either $uv$ or $vu$ is an arc). While \textsc{Graph Isomorphism} on tournaments may be easier than the general problem (we refer to~\cite{Schweitzer17,Wagner07} for the details), still it is unknown whether this special case can be solved in polynomial time and the best known algorithm is the quasi-polynomial algorithm of Babai~\cite{Babai16}.

\begin{figure}[ht]
\centering
\scalebox{0.7}{
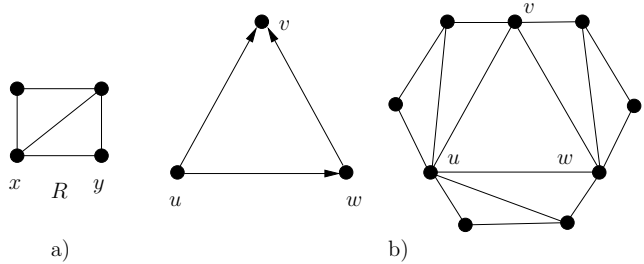}
\caption{Construction of $G$ from a tournament.}
\label{fig:turn}
\end{figure}

Let $T$ be a tournament. We construct the undirected graph $G(T)$:
\begin{itemize}
\item construct a copy of $V(T)$,
\item for every arc $uv$ of $T$, construct a copy of the graph $R$ shown in Figure~\ref{fig:turn} a) and identify the vertex $x$ with $u$ and $y$ with $v$ in the copy of $V(T)$ (see Figure~\ref{fig:turn} b)).
\end{itemize}
If $T_1$ and $T_2$ are $n$-vertex tournaments with $n\geq 2$, then is is straightforward to verify that $G(T_1)$ and $G(T_2)$ are 2-isomorphic and it is easy to construct their 2-isomorphism. However, $G(T_1)$ and $G(T_2)$ are isomorphic if and only if $T_1$ and $T_2$ are isomorphic.

Given this observation, it is natural to consider \probUWS on graph classes for which \textsc{Graph Isomorphism} is polynomially solvable. For example, what can be said about \probUWS on planar graphs?

The relation between Whitney switches and sorting by reversals together with the reduction in the proof of Proposition~\ref{prop:hard-unlabeled} indicates that as the first step, it could be reasonable to investigate the following problem for sequences that generalizes \probRS for permutations. Let $\pi=(\pi_1,\ldots,\pi_n)$ be a sequence of positive integers; note that now some elements of $\pi$ may be the same. For $1\leq i<j\leq n$, we define the \emph{reversal} $\rho(i,j)$ in exactly the same way as for permutations. Then we can define the \emph{reversal distance} between two $n$-element sequences such that the multisets of their  elements are the same; we assume that the distance is $+\infty$ if the multisets of elements are distinct. 

\defproblema{\probSRD}{Two $n$-element sequences $\pi$ and $\sigma$ of positive integers and a nonnegative integer $k$. }{Decide whether 
the reversal distance between $\pi$ and $\sigma$ is at most $k$.}

By the result of Caprara in~\cite{Caprara97}, this problem is \classNP-complete even if the input sequences are permutations. It is also known that the problem is \classNP--complete if the input sequences contain only two distinct elements~\cite{ChristieI01}. The question, whether \probSRD is \classFPT when parameterized by $k$, was explicitly stated in the survey of Bulteau et. al~\cite{BulteauHKN14} (in terms of strings) and is open and only some partial results are known~\cite{BulteauFK16}. We also can define the version of \probSRD for circular sequences and ask the same question about parameterized complexity. Using the idea behind the reduction in the proof of Proposition~\ref{prop:hard-unlabeled}, it is easy to observe that \probUWS on 2-connected series-parallel graphs is at least as hard as the circular variant of \probSRD.  
 
\paragraph{Acknowledgments.} We are grateful to Erlend Raa V{\aa}gset for fruitful discussions that initiated the research resulted in the paper.

\bibliographystyle{siam} 
\bibliography{switching}

\end{document}